\documentclass[a4paper,11pt,openright]{amsart}
\usepackage[utf8]{inputenc}
\usepackage{xcolor}
\usepackage{graphicx}
\usepackage{dcolumn}
\usepackage[hidelinks]{hyperref}
\usepackage[english]{babel}
\usepackage[autostyle, english = american]{csquotes}
\MakeOuterQuote{"}
\usepackage{amsthm}
\usepackage{mathrsfs}
\usepackage{amssymb}
\usepackage{bm}
\usepackage{physics}
\usepackage{cochineal}
\usepackage[cochineal]{newtxmath}

\usepackage[backend=bibtex,style=numeric,sorting=nyt,maxcitenames=5,mincitenames=3,maxbibnames=10,giveninits=true, doi = false, isbn = false, url=false]{biblatex}

\theoremstyle{plain}

\newtheorem{theorem}{Theorem}[section]
\newtheorem{proposition}[theorem]{Proposition}
\newtheorem{lemma}[theorem]{Lemma}

\theoremstyle{definition}
\newtheorem{definition}{Definition}[section]
\newtheorem*{remark}{Remark}

\DeclareMathOperator{\supp}{\text{supp}}
\DeclareMathOperator{\U''}{\mathit{U}^{(2)}_\mathit k}

\def\@tvsp{\mathchoice{{}\mkern-4.5mu}{{}\mkern-4.5mu}{{}\mkern-2.5mu}{}}
\def\ltrivert{\left|\@tvsp\left|\@tvsp\left|}
\def\rtrivert{\right|\@tvsp\right|\@tvsp\right|}
\newcommand\lnorm[1]{\ltrivert#1\rtrivert}
\newcommand\lonenorm[1]{\ltrivert#1\rtrivert_1}

\newcommand\ltwotwonormt[2]{{\ltrivert#1\rtrivert_{2,2}^{#2}}}

\begin{document}
\title{Local solutions of RG flow equations from the Nash-Moser theorem}
\author{Edoardo D'Angelo}
\address{Dipartimento di Matematica, 
Dipartimento di Eccellenza 2023-2027,
Università di Genova, Italy}
\address{Istituto Nazionale di Fisica Nucleare -- Sezione di Genova, Italy}
\email{edoardo.dangelo@edu.unige.it}
\author{Nicola Pinamonti}
\address{Dipartimento di Matematica, 
Dipartimento di Eccellenza 2023-2027,
Università di Genova, Italy}
\address{Istituto Nazionale di Fisica Nucleare -- Sezione di Genova, Italy}
\email{pinamont@dima.unige.it}
\date{}
\maketitle

\begin{abstract}
We prove local existence of solutions of a functional Renormalisation Group equation for the effective action of an interacting quantum field theory, when a suitable Local Potential Approximation is considered. To obtain this equation in a Lorentzian setting, a quantum state for the theory is selected, and a regulator consisting in a mass is added to the action. The flow equation for mass rescalings is then studied using the renown Nash-Moser theorem.
\end{abstract}

\section{Introduction}
% Less fancy incipit
%It is often said that every paper on General Relativity starts with a given metric, and those on Quantum Field Theory with a Lagrangian. Works based on the Renormalization Group (RG) approach often represent an exception to this folk rule. 

% Flashier incipit
In his historical essay on the development of the Standard Model \cite{Weinberg2018}, Weinberg recalls how Oppenheimer used to grumble that renormalization was just a way to sweep infinities under the rug.  In the last 70 years, from a collection of heuristic procedures and techniques to extract finite predictions from the infinities arising in Quantum Field Theory (QFT), the Renormalization Group (RG) has turned into a fundamental, organising principle of modern physics, describing the emergence of macroscopic phenomena from the interactions of microscopic degrees of freedom. In particular, in the Wilsonian picture \cite{Wilson73, Wilson1975}, renormalization describes how macroscopic effects arises from the microscopic degrees of freedom.

%In the Wilsonian picture of the RG \cite{Wilson73, Wilson1975}, the basic ingredients are the microscopic degrees of freedom of the system one is interested in, and their symmetries; all possible interactions are then organised in inverse powers of some energy scale $M$, provided by some dimensionful, effective coupling constant. 
%In principle, the degrees of freedom interact through \textit{all} possible terms allowed by symmetries; however, at a given energy scale $k$, only those terms whose inverse coupling is of order $k$ contribute, while all the others are suppressed by increasing powers of $M$. This suppression produces different effective field theories at different energy scales, and the RG governs the flow of all possible Lagrangians under changes in the energy scale $k$.

The Functional Renormalization Group (FRG) represents one of the modern implementations of the Wilsonian RG \cite{Berges2000, Dupuis2021, Pawlowski2005, Wetterich1992}. In this approach, correlation functions depends on the insertion of a regulator term that usually acts as a momentum-dependent mass $Q_k$, suppressing modes with frequencies higher than some energy scale $k$. The most convenient way of describing the effective theory at some scale $k$ is through the \textit{effective average action} $\Gamma_k$, which interpolates between the classical, microscopic action $I$ for $k\to \infty$, that is, when all quantum fluctuations are suppressed, and the full quantum action $\Gamma$ for $k \to 0$. The equation governing the flow of the effective average action under changes in the scale $k$, in Euclidean spaces, is the \textit{Wetterich equation} \cite{Wetterich1992, Morris1993}. 
%Therefore, the effective action is not a datum to use to extract predictions, but it is determined by the FRG flow as the solution of the Wetterich equation. The Wetterich equation in turn is a complicated functional differential equation, with initial data given by the effective average action at some scale $\Lambda$, usually the bare classical action for some theory.

The FRG has been successfully applied to many different physical situations, although mainly in Euclidean spaces: from condensed matter systems to high-energy physics, most notably QCD, (see for example \cite{Dupuis2021} and references therein) and quantum gravity \cite{Percaccibook2017, Reuter1996, ReuterSaueressig2012, ReuterSaueressig2019, Saueressig2023}, where it represents the principal tool of investigation of the asymptotic safety scenario. In fact, thanks to the structure of the Wetterich equation, the FRG admits non-perturbative approximation schemes that can go beyond usual perturbation theory, allowing the study of strongly coupled systems and perturbatively non-renormalisable theories such as quantum Einstein gravity.

The approximation schemes in the FRG usually start from an ansatz for the effective average action, based on different expansions of $\Gamma_k$ into polynomials of the fields and their derivatives; the most used is the Derivative Expansion (DE), based on the number of spatial derivatives of the fields. The lowest order approximation in DE is the Local Potential Approximation (LPA), in which $\Gamma_k$ contains only an effective potential term $U_k(\phi)$, containing no derivatives of the fields, and a kinetic term corresponding to the classical one.

According to the RG philosophy, every possible interaction term is admitted in principle in the effective average action. The FRG flow reflects this behaviour in its mathematical structure, which can be written as
\[
\partial_k \Gamma_k = \frac{i}{2}\langle \partial_k q_k , (\Gamma_k^{(2)} -q_k)^{-1} \rangle \ ,
\]
where $Q_k = - \frac{1}{2} \int q_k \chi^2$ is the regulator term acting as an artificial mass for the field $\chi$, and $q_k$ is its integral kernel. The pairing is the standard pairing of bi-distributions over the spacetime $\mathcal M$ namely on $\mathcal{M} \times \mathcal M$. In Euclidean spacetimes, where Euclidean invariance selects a natural notion of vacuum, (since the Euler-Lagrange equations for the action $I$ are elliptic, rather than hyperbolic), the inverse $(\Gamma_k^{(2)} -q_k)^{-1}$ is unique. However, due to the appearance of the second derivative of $\Gamma_k$ on the r.h.s, we observe that, independently from the initial datum for $\Gamma_k$, the flow will always produce additional interaction terms in the effective average action. Thinking of the inverse in a perturbative expansion as a Neumann series, we observe that any power of some field polynomial can also appear along the flow. Only the artificial truncation of $\Gamma_k$ in a polynomial expression of finite order prevents the generation of infinite terms.

A standard approach in the literature is to simply truncate the expansion of $\Gamma_k$ in a finite number of terms; the Wetterich equation then reduces 
%from a functional differential equation in
to a system of coupled, partial differential equations for the coefficients of the field polynomials, obtaining the $\beta-$\textit{functions} for the theory. However, the obtained solutions are not solutions of the full equation, because of the truncation. Furthermore, little control on the quality of the approximation scheme, compared to the full theory space, is possible. 

Mathematically, the problem of the generation of every possible term along the flow is connected with the problem of \textit{loss of derivatives}: intuitively, since the r.h.s of the Wetterich depends on the inverse of the \textit{second derivative} of the effective average action, a Green operator (a fundamental solution) for the Wetterich equation will also depend on the second derivative $\Gamma^{(2)}_k$. It follows that, if the Wetterich equation is an operator acting on some space of $C^n$ functions of the fields, its solutions will generally be only $C^{n-2}-$regular, losing two derivatives. Due to the loss of derivatives, standard iterative procedures to produce solutions in suitable Banach spaces fail to converge.

Recently, together with Drago and Rejzner we developed a new approach to the FRG, based on the methods of perturbative Algebraic Quantum Field Theory (pAQFT) \cite{AAQFT15, Brunetti2009, Fewster2019, Rejzner2016}, 
which was later extended to the case of gauge theories by one of the authors and K. Rejzner 
\cite{DDPR2022, DR2023}. The approach is fully Lorentzian, and allows for a generalization of the Wetterich equation to generic states and curved backgrounds, where a natural notion of vacuum is usually not at disposal. These Lorentzian RG flow equations exhibit  a state dependent flow, and are based on a Hadamard regularisation of the UV divergences, instead of a regularisation based on a momentum-dependent regulator $Q_k$. The Hadamard regularisation is made possible by selecting the class of Hadamard states as background states for the free theory. Both these features are shared with a similar approach developed recently for scalar fields in cosmological spacetimes \cite{Banerjee2022}. On the other hand, the regulator term is chosen as an artificial mass term, without momentum dependence: such a regulator term is particularly useful to preserve unitarity of the $S-$matrix and Lorentz invariance of the flow, and do not modify the structure of the propagator. This regulator term has been called \textit{Callan-Symanzik-type cut-off}, and the respective RG flow equation is also known as functional Callan-Symanzik equations \cite{Alexandre2001}. 

The RG flow equation then takes the expression of a functional differential equation for the effective potential $U_k$, depending on: a fixed background geometry $\mathcal M$; the smooth part of the two-point function $w$ of a reference state $\omega$ for the free theory, defined by a quadratic action $I_0$; the advanced and retarded propagators $\Delta_{A,R}^U$ for the wave operator $P_0  + \U''$, where $P_0 = I^{(2)}_0$ is the wave operator for the free theory; and some initial value $U_{k= \Lambda}$. We derive its form in \eqref{eq:RG for U}
and we recall here its form,
\begin{equation}\label{eq:RG for U intro}
\partial_k U_k = - \frac{1}{2}  \int_{\mathcal{M}}  \partial_k q_k(x) (1 -\Delta_R^U \U'') w (1 - \U'' \Delta_A^U )(x,x) \dd\mu_x\ .
\end{equation}
%Notice that, in general, the r.h.s is only defined as a formal power series, since the advanced and retarded propagators for the quantum wave operator $\Delta^U_{A,R}$ are not guaranteed to exist non-perturbatively. However, 
Equation \eqref{eq:RG for U intro} is non-perturbative in the coupling constant, as the Wetterich equation in Euclidean space, and so it allows for non-perturbative approximation schemes.

In this paper, we take a step further in clarifying the mathematical structure of the RG flow equations, and we prove that, with possibly a non-polynomial effective potential $U_k$ that contains no derivatives of the Dirac delta or of the field, the RG flow equations admit a local solution for suitably regular initial conditions.

In order to prove the existence of local solutions for the RG flow \eqref{eq:RG for U}, we need to choose an appropriate approximation. Inspired by Euclidean FRG approaches, as a first step towards more general results we choose to approximate $U_k$ with the Local Potential Approximation, as a local function of the field $\phi$ with no derivatives given \eqref{eq:LPA}, 
\[
U_k(\phi) = \int_{\mathcal{M}} u(\phi(x), k) f(x) \dd\mu_x  \ ,
\qquad 
U_k^{(2)}(\phi)(x,y) = \partial_\phi^2 u(\phi, k) f(x) \delta(x,y)  \ ,
\]
where $f$ is a compactly supported smooth function which is equal to $1$ on large regions of the studied spacetime. 
We further assume that the field $\phi$ is constant over the whole space, so that $\partial^2_\phi u$ is function of $k$ and $\phi$ only. However, notice that $U_k$ can be any smooth non-polynomial function of the field $\phi$.

Within this approximation, the r.h.s. of the RG flow equation can be written in terms of the map given in \eqref{eq:G} which we recall here
\[
G_k(\partial_\phi^2u) :=- \frac{1}{2 \lonenorm{f}}  \int_{\mathcal{M}}  \partial_k q_k(x)
\left\{(1-\partial_\phi^2 u \Delta_R^u f) \otimes (1-\partial_\phi^2 u \Delta_R^u f)(w)(x,x)
\right\} \dd\mu_x
\ .
\]
The RG flow equation reduces to an equation for $u(\phi,k)$. We are thus interested in studying the existence of solution for the following problem
\begin{equation}
\begin{cases} 
\partial_k u = G_k(\partial_\phi^2u) \ , \\
u(\phi, a) = \psi \ , \\
\left.u\right|_{\partial X \times [a,b]} = \beta \ 
\end{cases}
\end{equation}
where $\psi$ and $\beta$ are known functions which characterise respectively  the initial and boundary conditions of the problem. 

The LPA greatly simplifies the RG flow equation, which now is an equation for $u$, as a function of $k$ and $\phi$. However, the LPA does not simplify the problem of the loss of derivatives. 

The main result of this paper is the proof of local existence of solutions of this problem which is given below in Theorem \ref{theorem:existence of solutions}. The proof of local existence of local solutions of this problem is an application of the renown Nash-Moser Theorem.
 
Nash provided a beautiful theorem to prove local existence of solution of non-linear partial differential equations in spaces of smooth functions, which are particularly suited to deal with the problem of loss of derivatives \cite{Nash1956}. The theory was first developed in the context of isometric embeddings of Riemannian manifolds by Nash, and then further generalised by Moser \cite{Moser1966-I,Moser1966-II}. Hamilton \cite{Hamilton1982} provided a particularly natural setting for the theorem in the space of tame Fréchet spaces.

\subsection{Nash-Moser theorem in the Hamilton formulation} \label{sec:hamilton}

In the formulation of Hamilton, the Nash-Moser theorem is given for elements in a suitable tame Fréchet space. Mainly to fix notation, we recall here some basic definitions and the Theorem we shall use to get the main result presented in this paper.
\begin{definition}
A \textit{seminorm} on a vector space $F$ is a function $\norm{\cdot} : F \to \mathbb R$ such that, $\forall \ f, \ g \in F$ and $\forall c \in \mathbb R$: (i) $\norm{f} \geq 0$; 
(ii) $\norm{f+g} \leq \norm{f} + \norm{g}$; 
(iii) $  \norm{cf} = \abs{c} \norm{f} $.
% \end{enumerate}
% \begin{enumerate}
% \item $\norm{f} \geq 0$; \\
% \item $\norm{f+g} \leq \norm{f} + \norm{g}$; \\
% \item $  \norm{cf} = \abs{c} \norm{f} $ \\
% \end{enumerate}
A collection of seminorms $\{ \norm{\cdot}_n \}_{n \in \mathbb N}$ defines a unique topology such that a sequence $f_i \to f \Leftrightarrow \norm{f_i - f}_n \to 0 \ \forall n \in \mathbb N$.
A \textit{locally convex topological vector space} is a vector space with a topology arising from a collection of seminorms. The topology is called \textit{Hausdorff} if $f=0$ when $\norm{f}_n = 0 \ \forall n$. The topology is called \textit{metrizable} if the family $\{ \norm{\cdot}_n \}_n$ is countable, 
%A \textit{Cauchy sequence} is a sequence such that $\norm{f_i - f_j}_n \to 0$ for $i, j  \to \infty \ \forall n$, 
and the space $F$ is \textit{complete} if every Cauchy sequence converges.
A \textit{Fréchet space} is a complete Hausdorff metrizable locally convex topological vector space, and a graded Fréchet space has a collection of seminorms that are increasing in strength, so that $\norm{f}_n \leq \norm{f}_{n+1} \ \forall n$.
\end{definition}

\begin{definition}
A graded space $F$ is \textit{tame} if, given the space $\Sigma(B)$ of exponentially decreasing sequences in some Banach space $B$, it is possible to find two linear maps $L: F \to \Sigma(B)$, $M: \Sigma(B) \to F$, such that $ML : F \to F$ is the identity
\begin{equation}
F \to^L \Sigma(B) \to^M F \ .
\end{equation}

Consider two graded spaces $F$ and $G$, and a map $P: \mathcal U \subset F \to G$ from an open subset $\mathcal U$ of $F$ to $\mathcal G$. The map $P$ is \textit{tame} of degree $r$ and base $b$ if it is continuous and satisfies
\[
\norm{P(f)}_n \leq C ( 1 + \norm{f}_{n+r})
\]
for all $f$ in the neighbourhood of each $f_0 \in \mathcal U$, for all $n \geq b$, and with a constant $C$ that may depend on $n$.
\end{definition}

In this setting, we can make use of the following, classic theorem on the inverse function problem:

\begin{theorem}[Nash-Moser theorem in Hamilton's formulation]
Consider a smooth tame map $P: \mathcal U \subset F \to G$ between two tame Fréchet spaces $F$ and $G$. Suppose that 
\begin{enumerate}
    \item the linear map $DP(u) v = f$ obtained as the first functional derivative of $P$ has unique inverse $E(u)f = v \ \forall \ u \in \mathcal U$ and all $f\in G$, and
    \item the inverse map $E: \mathcal U \times G \to F$ is smooth tame.
\end{enumerate} 
Then $P$ is locally invertible and $P^{-1}$ is a smooth tame map.
\end{theorem}

\subsection{Strategy and summary of results}

In order to prove the theorem, using Hamilton's formulation of Nash-Moser theorem, the RG flow equation needs to satisfy a number of assumptions. First of all, it must be cast in the form of a suitable map acting on a tame Fréchet space. Requiring that $\phi$ and $k$ are limited in some compact interval, and that $u_k$ is a smooth function, is sufficient for $u_k$ to be an element of a tame Fréchet space. Secondly, the operator $\mathcal{RG}: u \in F_0 \to F$ defining the RG flow equation (see Def. \eqref{def: RG operator}) must be a smooth tame map between tame Fréchet spaces. In order to be tame, the RG operator must satisfy some estimates on its seminorms. Assuming that $u$ lies in some neighbourhood of $0$ (by requiring that a suitable seminorm of $u$, given below in Eq. \eqref{eq:seminorm}, is $\|u\|_4 < A$ for sufficiently small $A$), it is possible to prove these estimates using the Gr\"onwall lemma, since the normal-ordered interacting propagator $G_k(\partial^2_\phi u)$ satisfies a recursive integral inequality. 
%Loosely speaking, the inequality follows from its definition as a perturbation series in terms of the free propagators and the smooth part of the state two-point function.

Then, the linearisation of the RG operator must be an invertible smooth tame operator, and its inverse must be tame smooth. In the LPA, the linearisation $L = D \mathcal{RG}$ takes the form of a parabolic equation, analogous to a heat equation with a $k,\phi-$dependent heat conductivity $\sigma$. The inverse of linear parabolic equations is known \cite{Friedman} (see also \cite{Dappiaggi2020, Murro2022}, and the inverse of the linearised RG operator can be constructed from the heat kernel. Once the inverse of the linearised RG operator is known, it is possible to prove that it is tame smooth.

All these results are presented and proved in Propositions \ref{prop: RG tame smooth}, \ref{prop:linearised RG operator tame smooth}, \ref{prop: existence inverse E}, and \ref{prop: inverse is tame smooth}. These are used to prove our main result, Theorem \ref{theorem:existence of solutions}, on the existence of local solutions of the RG flow which we report here for completeness in a compact form.
%\begin{theorem}
%It follows from the Nash-Moser theorem that the RG operator admits a %unique family of tame smooth local inverses. This guarantees the %existence of local solutions of the RG flow equations.
%\end{theorem}
%The solutions to the RG flow equation $\mathcal{RG}(u)=0$ are simply given by $u= \mathcal{RG}^{-1}(0)$.
\begin{theorem}
The RG operator admits a unique family of tame smooth local inverses, and unique local solutions of the RG flow equations exist.
\end{theorem}

% \begin{theorem}
% Unique local solutions of the RG flow equations exist.
% \end{theorem}

The material is organised as follows. We start with a review of the derivation of the RG flow equations, to set the notation and clarify the underlying framework, in section \ref{sec:FRG equations}. This section closely follows the presentation in \cite{DDPR2022}. In section \ref{sec:RG of the effective potential}, we clarify the structure of the RG flow equation, expanding on some points the derivation presented in \cite{DDPR2022}, and we show how to compute the interacting propagator $G_k$ from the free propagators and the underlying state for the free theory, to get the RG flow equations as closed differential equations for the effective potential $U_k$.

In the main part, starting from section \ref{sec:existence of local solutions}, we prove our main theorem. First of all, we define the LPA and identify the appropriate tame Fréchet spaces in which we want to solve the equation. We then proceed proving the main propositions \ref{prop: RG tame smooth}, \ref{prop:linearised RG operator tame smooth}, \ref{prop: existence inverse E}, and \ref{prop: inverse is tame smooth}. In theorem \ref{theorem:existence of solutions} we state our main result on the existence of local solutions of the RG flow equations, which follows immediately from the propositions.

\section{Functional RG flow equations}\label{sec:FRG equations}
\subsection{Generating functionals}

%\color{blue}
We recall here the main steps in the derivation of the RG flow equation in the form given in \eqref{eq:RG for U intro} presented in \cite{DDPR2022}.
We refer to the paper \cite{DDPR2022} for further details.

The methods used to represent the objects we are working with are those proper of pAQFT \cite{AAQFT15, Brunetti2009, Brunetti1999, Brunetti2001, HollandsWald2001a, HollandsWald2001b, Rejzner2016}. In this framework field observables are seen as functional over smooth field configurations. The quantum properties manifest themselves in the various products used to multiply those objects and in the involutions used to construct positive elements. In this way one obtains a $*$-algebra $\mathcal{A}$ of field observables. 
Expectation values of observables are obtained testing elements of $\mathcal{A}$ over a positive, normalised linear functional $\omega$.
We refer to \cite{Rejzner2016} for full details on the quantization procedure. 

%In this section we briefly recall the important notions to perturbatively define an interacting scalar field theory in the algebraic setting. We refer to \cite{Rejzner2016} for further details.
We work with a globally hyperbolic spacetime $\mathcal{M}$, which is a smooth, oriented, and time oriented manifold, equipped with Lorentzian metric $\eta$ which makes $\mathcal M$ globally hyperbolic.
For simplicity we shall assume that $\mathcal{M}$ is stationary and ultra-static. This restriction is required in order to have good control on the regularity of the advanced and retarded propagators; however, the estimates that we prove are known to hold in more general spacetimes, such as de Sitter.

%We start from given classical data, of a Lorentzian, globally hyperbolic manifold $\mathcal M$ with metric $g$ and the classical off-shell configuration space $\mathscr E = C^\infty(\mathcal M) \ni \chi$, consisting of smooth sections on the manifold. The \textit{classical algebra of observables} consists of the closure with respect to the pointwise product of the space of complex-valued functionals of off-shell field configurations, $F \in \mathscr F(\mathscr E)$ with compact support, i.e.
%

We start with the action of a quantum field theory propagating on $\mathcal{M}$. For simplicity we discuss here the scalar case. The lagrangian density of the theory is
\[
\mathcal{L}(\chi) = \mathcal{L}_0(\chi)+\mathcal{L}_I(\chi) = -\frac{1}{2} (g^{-1}(\partial \chi, \partial \chi) + \xi \chi^2 R + m^2 \chi^2) +  \mathcal{L}_I(\chi)
\]
%where $f$ is a cutoff which is equal to $1$ in the region of the spacetime $\mathcal{M}$ where we are analyzing our theory and it is inserted to make the integral over the manifold finite.
where $\chi$ is a field configuration and
$\mathcal{L}_0(\chi)$ is the Lagrangian density of the free theory, quadratic in the fields. The free theory can be quantised, providing a $*-$algebra of observables for which states are known to exist. $\mathcal{L}_I$ is the interaction Lagrangian, the contribution to $\mathcal{L}$ which is more than quadratic in the fields.
We denote by $I$ and $I_0$ the bare action associated to the Lagrangian density $\mathcal{L}$ and $\mathcal{L}_0$ respectively, and by
\[
V(\chi) := \int_{\mathcal{M}} \lambda \mathcal{L}_I(\chi) f \dd\mu_x
\]  
the interacting action. $f$ is a cut-off function, which is equal to $1$ in the region of the spacetime $\mathcal{M}$ where we are analysing our theory, and it is inserted to make the integral over the manifold finite. Similar regulators needs to be considered also in $I$ and $I_0$.
% $V$ is the smeared interaction Lagrangian.
%Both $I_0$ and $V$ are seen as functional over the classical field configuration $\chi$.  
%
Since a direct construction of the observables of the interacting theory is not available, we use perturbative methods to represent interacting fields over the free theory. Interacting fields are thus represented as formal power series in the parameter $\lambda$, which governs the non-linear coupling of the theory. The coefficients of the formal power series are elements of the free theory.
See \cite{Brunetti2009, Brunetti1999, HollandsWald2001a, HollandsWald2001b, HollandsWald2004} for further details.
States are thus linear functionals on the algebra of the free theory.
%\color{black}

We now need to introduce the generating functionals for correlation functions, which are the starting point of most treatments of the functional Renormalization Group (fRG) \cite{Berges2000, Niedermaier2006,  Percaccibook2017, ReuterSaueressig2012, ReuterSaueressig2019, Saueressig2023}.

%\color{blue}
In this framework, the generating functional of correlation functions depends on the state of the free theory, and it is defined as
%,
%\[
%Z(j) = \int \mathcal D \chi e^{iI(\chi)} \ ,
%\]
%where $I(\chi)$ is the bare action of the theory. As it is known, however, the path-integral is not well-defined in general Lorentzian spacetimes, and its interpretation as the generating functional for correlation functions in generic states and in possibly curved, Lorentzian spacetimes is particularly problematic. For these reasons, we replace the above definition with
\[
Z(j) := \omega \left ( S(V)^{-1} \star S(V+J) \right ) \ ,
\]
where $J = \int j(x) \chi(x) \dd x$.
%\color{blue}
$V=I-I_0$ is the interaction action, and $S(V)$ is the time ordered exponential of $V$. The $\star-$product represents the quantum, non-commutative product in the algebra of the free theory constructed from $I_0$, and  $\omega$ is the state in which we are interested to evaluate correlation functions.
% that theory.  
%\color{black}

%For vanishing sources we have $Z(0) =1$. It coincides by definition with the Bogoliubov map applied on the time-ordered exponential of the source term,
%\[
%Z(j) = \omega \left( R_V(S(J)) \right ) \ .
%\]
The functional derivatives of $Z(j)$ for vanishing sources gives the interacting, time-ordered correlation functions of the interacting fields,
\[
\eval{Z^{(n)}(j)}_{j=0} = i^n \omega \left ( S(V)^{-1} \star (S(V) \cdot_T \chi(x_1) \cdot_T ... \cdot_T \chi(x_n) \right )).
\]
%\color{blue}
We refer to \cite{DDPR2022} for a discussion of the relation to the standard approach present in the physics literature.
However, we observe that the standard approach may be recovered only for states $\omega$ for which the Gell-Mann-Low formula holds, namely when 
the star product above factorises in the product of expectation values. Equilibrium states at finite temperature or states on curved backgrounds do not have this property.  
%\color{black}

%When the star product factorizes in the state functional, that is, whenever the Gell-Mann-Low formula holds, the above reduces to the standard expression for correlation functions for non-vanishing sources
%\begin{align*}
%\langle \chi(x_1) \cdot_T...\cdot_T \chi(x_n) \rangle_j 
%&= \omega \left ( S(V)^{-1} \star S(V+J) \cdot_T \chi(x_1)... \cdot_T \chi(x_n)  \right ) \\
%&= \frac{\omega \left (  S(V+J) \cdot_T \chi(x_1) ... \cdot_T \chi(x_n)  \right )}{\omega(S(V))} \\
%&''='' \frac{ \int \mathcal D \chi e^{i(I+J)} \chi(x_1)...\chi(x_n)}{\int \mathcal D \chi e^{iI}} \ ,
%\end{align*}
%where the heuristic map between the physicist's notation and the present notation is given by
%\[
%\int \mathcal D \chi F \leftrightsquigarrow \omega (T(F)) \ .
%\]
%
The functional Renormalization Group approach works by deforming the underlying theory with an artificial mass scale $k$. In standard treatments, the procedure consists in the addition of a non-local regulator, quadratic in the fields, to the bare action $I$, acting as a scale-dependent mass term.
Although a non-local, momentum dependent regulator implements a Wilsonian renormalization flow, (i.e., a genuine coarse-graining procedure in which field modes with increasing frequency are progressively integrated out), at least in Euclidean settings, in the case of Lorentzian signature it appears less favourable, due to its non-local nature in position \cite{DDPR2022}. In particular, it can spoil the unitarity of the $S-$matrix and it can introduce artificial poles in the propagator. For this reason, in \cite{DDPR2022} we chose to use a local regulator term
\[
Q_k(\chi) = - \frac{1}{2} \int \dd\mu_x q_k(x) T \chi^2(x) \ .
\]
Notice that the time-ordering operator $T\chi^2$ naturally introduces a normal-ordering prescription, so that $T \chi^2$ is actually finite.

This regulator function will act as an artificial mass term in the correlation functions. Therefore, although it does not regularise UV divergences, it regularises IR divergences. In turn, it does not spoil the unitarity of the $S-$matrix and does not alter the structure of the propagator. Moreover, in \cite{DDPR2022} it was proven that, in the limit of infinite mass, $k\to \infty$, the Feynman propagator reduces to zero and quantum effects are completely suppressed. We then see that with the introduction of such a term, we can picture the flow of correlation functions under changes of the scale $k$, from large scales to the vanishing limit $k \to 0$, as a flow from the classical theory to the quantum one. This in turn justifies the terminology of \textit{renormalization group flow}. Finally, notice that local regulators of this type were already introduced in \cite{Alexandre2001}, where the RG flow equations with a local regulator have been called \textit{functional Callan-Symanzik equations}, and appear as a special case of the Wetterich equation \cite{Wetterich1992} in the case of local regulators instead of non-local ones. More recently, local regulators have been used in the Lorentzian setting in \cite{Fehre2021}, where they are used to study the flow of the graviton spectral functions, and in \cite{Litim2001, Litim2006} in the context of renormalization of thermal field theories.

%\color{blue}
We thus deform 
%Instead of adding the regulator directly to the action, we choose to deform 
%\color{black}
the generating functional $Z$, defining a $k-$dependent generating functional
\[
Z_k(j) = \omega \left ( S(V)^{-1} \star S(V+J + Q_k) \right ) \ .
\]
%In our previous work \cite{DDPR2022}, we discussed how such a regulator term in the generating functional effectively acts as a $k-$dependent massive term in the correlation functions, and so acting as an IR regulator for massless theories. 
%At the same time, being a non-trivial deformation of the original theory, it gives us more freedom in the renormalization ambiguity, which can be exploited, at least in some cases, to remove the renormalization parameters in the $\beta-$functions for the couplings, leading to a regularization independent RG flow.

From the definition of the regularised generating functional $Z_k$, the steps to define the effective action are standard: we first define the (regularised) generating functional for the connected correlation functions as
\[
e^{i W_k(j)} := Z_k(j) \ .
\]
The first derivative of $W_k$ defines the \textit{classical field} $\phi$ as a function of $j$,
\[
\phi_j(x) := W^{(1)}_k(j)(x) = \frac{\delta W_k(j)}{\delta j(x)} = e^{-iW_k} \omega \left(S(V)^{-1} \star S(V+Q_k+J) \cdot_T \chi(x) \right) = \langle \chi \rangle \ ,
\]
while the second derivative is proportional to the connected, interacting Feynman propagator
\begin{equation}\label{eq:W(2)}
- i W^{(2)}_k(j) = \langle \chi(x) \cdot_T \chi(y) \rangle - \phi_j(x) \phi_j(y) \ .
\end{equation}
In the above relations, we introduced the angle brackets to denote the weighted expectation value of an interacting operator $F$, for non-vanishing sources and regulator:
\[
\langle F \rangle := e^{-iW_k} \omega \left(S(V)^{-1} \star S(V+Q_k+J) \cdot_T F \right)
\]

The relation between $\phi$ and $j$ can be inverted, giving
\begin{equation} \label{eq:connection between j and phi}
- j_\phi(x) = (P_0 \phi)(x)
    +Q^{(1)}_k(\phi)(x)
    +\langle T V^{(1)}(\chi) \rangle \ ,
\end{equation}
which can be solved perturbatively at all orders \cite{DDPR2022}. Thanks to this inversion, we can define the Legendre transform of $W_k$ as
\begin{equation}\label{eq:gamma-W}
\tilde \Gamma_k(\phi) := W_k(j_\phi) - J_\phi(\phi) \ ,
\end{equation}
satisfying
\[
\tilde \Gamma^{(1)}_k(\phi) = - j_\phi \ .
\]
The second derivative of $\tilde \Gamma_k$ and $W_k$ are related by the standard formula between Legendre transforms,
\begin{equation} \label{eq:Gamma and W second derivatives}
 ( \Gamma_k^{(2)} - q_k) (x,z) W^{(2)}_k(z,y) = - \delta(x,y) \ .
\end{equation}
Therefore, $W^{(2)}_k$ is a propagator for the \textit{quantum wave operator} $ \Gamma_k^{(2)} - q_k$.
%, we can subtract its divergences defining its normal-ordered counterpart $:W^{(2)}_k:_{\tilde H_F}$ that satisfies the wave equation,
%\begin{equation} \label{eq:quantum wave operator regularised propagator}
%( \Gamma_k^{(2)} - q_k) :W^{(2)}_k:_{\tilde H_F} = 0 \ .
%\end{equation}
%For now, we leave the counterterm $\tilde H_F$ implicit; we will discuss its explicit expression in the next sections.

Finally, the \textit{effective average action} is defined subtracting the classical term $Q_k(\phi)$ from the Legendre transform of $W_k$:
\[
\Gamma_k(\phi) := \tilde \Gamma_k(\phi) - Q_k(\phi) = W_k(j_\phi) - J_\phi(\phi) - Q_k(\phi) \ .
\]

\subsection{RG flow equations}
The RG flow equations govern the flow of the effective average action under scaling of the parameter $k$. We now remember the main steps in their derivation.

The first step is computing the $k-$derivative of $Z_k$, which is straightforward from its definition
\[
\partial_k Z_k(j) = i \omega \left ( S(V)^{-1} \star S(V+Q_k + J) \cdot_T \partial_k Q_k \right) \ .
\]
The $k$-derivative of $W_k$ follows immediately,
\begin{equation}\label{eq:equationWk-Q}
\partial_k W_k(j) = \langle \partial_k Q_k \rangle \ ,
\end{equation}
%\color{blue}
where we recall that
\[
\partial_kQ_k(\chi) = - \frac{1}{2} \int \dd x \partial_k q_k(x) T \chi^2(x) \ .
\]
%\color{black}
Notice that, thanks to normal ordering introduced by the $T-$products of local observables, the flow equation is UV finite. 

%\nicomment{The contribution $\langle T \chi^2 \rangle$ in $\langle \partial_k Q_k \rangle$ can be obtained as
%\[
%\langle T \chi^2(x) \rangle =: \lim_{y \to x} %\left(\langle  \chi(x) \cdot_T \chi(y) \rangle - %\tilde{H}_F(x,y)\right) \ ,
%\]
%where the counterterms $\tilde H_F(x,y)$ are implicitly defined by a re-summation of the divergences at each order in a perturbative expansion, and they make the last expression UV finite.
%}

%\color{blue}
The contribution $\langle T \chi^2 \rangle$ in $\langle \partial_k Q_k \rangle$ can be obtained as
\[
\langle T \chi^2(x) \rangle   = \lim_{y \to x} \left(\langle  \chi(x) \cdot_T \chi(y) \rangle - \tilde{H}_F(x,y)\right) \ ,
\]
where the counterterms $\tilde H_F(x,y)$, arising from the expectation value of a normal-ordered quantity, implements normal ordering in the interacting theory and make the expression finite.
%and it is necessary to make the expression finite.

%\color{black}
%The expectation value of $T \chi^2$ can be written in two, equivalent ways, thanks to the relation
%\begin{equation}
%T \chi^2(x)  = \lim_{y \to x} \chi(x) \cdot_T \chi(y) - H_F = \lim_{y \to x} \chi(x) \star \chi(y) - (H + \frac{i}{2} \Delta) \ ,
%\end{equation}
%where we recall that $H_F = H + i \Delta_A$.
%
%{\color{red} spiegare i propagatori
%
%recalling \eqref{eq:W(2)}
%}
%
%The RG flow equation is obtained by first rewriting the flow equation for $W_k$ as 
%\begin{equation}\label{Eq: W-RG flow equation}
%\partial_k W_k= - \frac{1}{2} \int \dd x \partial_k q_k(x) \langle \lim_{y \to x} \chi(x) \cdot_T \chi(y) - H_F(x,y) \rangle \ .
%\end{equation}
%
%
%By implicitly defining  the counterterms $\tilde H_F(x,y)$, we can rewrite the above term as (the Lorentzian generalization of ) the \textit{Polchinski equation} \cite{Polchinski1983}, 
%
%\color{blue}
Recalling \eqref{eq:W(2)}, we can rewrite \eqref{eq:equationWk-Q} as (the Lorentzian generalization of ) the \textit{Polchinski equation} \cite{Polchinski1983}. 
%\color{black}
%
%noticing that the second derivative of $W_k$ is the connected 2-point function:
\begin{equation} \label{eq:polchinski}
\partial_k W_k= - \frac{1}{2} \lim_{y \to x} \int \dd x \partial_k q_k(x) \left [-i W^{(2)}_k(x,y) + \phi(x) \phi(y) -  \tilde H_F(x,y)  \right ] \ .
\end{equation}
%\color{blue}
Notice that $-i W^{(2)}_k$ is the propagator of the interacting theory. In the case of fundamental solutions of free hyperbolic equations, 
the counter-terms $\tilde{H}_F$ necessary to implement normal ordering are well known, and are given in terms of suitable Hadamard parametrix $H_F$
(see e.g. \cite{KayWald}). This normal-ordering procedure is known as Hadamard subtraction, or point-splitting regularisation. 
We refer to \cite{KayWald} for further details.
%\color{black}

Recalling \eqref{eq:gamma-W}, the RG flow equation \eqref{eq:polchinski} can now be written as the flow equation for $\Gamma_k$ with a self-consistency relation
%{\color{red} controllare le $i$ e i segni}
%\color{blue}
%\begin{align} 
%\label{eq:RG flow}
\begin{gather*}
\partial_k \Gamma_k = -\frac{1}{2} \int_x \partial_k q_k(x) :G_k(x,x):_{\tilde H_F} \\
(\Gamma^{(2)}_k - q_k) G_k  = i \delta \ ,
\end{gather*}
%\end{align}
while the normal-ordering prescription is given by
\[
:G_k(x,y):_{\tilde H_F} = G_k(x,y) - \tilde H_F(x,y) 
\]
and $G_k(x,y) = -i W^{(2)}_k(x,y)$.

In the following section we will gain more insight in the r.h.s of the RG flow equation. 
In particular, we will use the fact that 
$G_k$ is a  fundamental solution of $(\Gamma^{(2)}_k - q_k)$ and $\tilde{H}_F$ is a fundamental solution of $(\Gamma^{(2)}_k - q_k)$ up to known smooth terms.
%
%
%both $G_k$ and $\tilde{H}_F$ are fundamental solutions of $(\Gamma^{(2)}_k - q_k)$, and thus 
Hence, $:G_k:$ is a bi-solution of the equation of motion up to known smooth terms.
Moreover, in the regions of spacetime $J^-(\mathcal{O})$ where $V\to 0$ and $q_k\to 0$, 
$G_k$ reduces to $\omega_2+i\Delta_A$, the Feynman propagator of the free theory, where $\omega_2$ is the two-point function of the free theory. Similarly, $\tilde{H}_F$ reduces to $H_F$,  the Hadamard parametrix of the free theory. These observations allow to obtain an explicit form of $:G_k:(x,y)$ in terms of the effective average action and the smooth part of the state for the free theory, $w = \omega_2+i\Delta_A-H_F$, by means of the classical  M{\o}ller maps \cite{Drago2015}.

\section{Quantum equations of motion and RG flow for the effective potential}
\label{sec:RG of the effective potential}
Recalling that, by definition,
\[
\frac{\delta}{\delta \phi(x)} \left (\Gamma_k + Q_k \right ) = - j_\phi(x) \ ,
\]
equation \eqref{eq:connection between j and phi} can be recast in the \textit{quantum equation of motion} (QEOM)
\begin{equation} \label{eq:QEOM}
\frac{\delta}{\delta \phi(x)} \left (\Gamma_k + Q_k \right ) = P_0 \phi(x) + \langle V^{(1)}(x) + Q_k^{(1)}(x) \rangle \ .
\end{equation}
The above relation can also be re-expressed as a Dyson-Schwinger equation, since $P_0 \phi  = \langle \frac{\delta I_0}{\delta \chi} \rangle$, so that
\begin{equation} \label{eq:DSE}
\frac{\delta \Gamma_k}{\delta \phi(x)} = \langle \frac{\delta I}{\delta \chi(x)} \rangle \ .
\end{equation}
The QEOM suggests to decompose the effective average action into
\begin{equation} \label{eq:decomposition effective average action}
\Gamma_k^{(2)}(\phi) - q_k := P_0  + U^{(2)}_k(\phi) \ ,
\end{equation}
where $P_0 := I^{(2)}_0$ is the Green hyperbolic operator defined by the free action, while the {effective potential} $U_k(\phi)$, is defined by the relation
\begin{equation} \label{eq:relation U-V}
U_k^{(1)}(x) = \langle V^{(1)}(x) + Q_k^{(1)}(x) \rangle \ .
\end{equation}
The effective potential includes all the quantum corrections to the interaction $V$, and it can be seen as a non-perturbative definition for the sum of perturbative Feynman diagrams. In its perturbative expansion, the effective potential contains non-localities and possibly higher-derivative terms. 

We call the operator $P_0 + \U''$ the \textit{quantum wave operator}. In terms of the effective potential, the relation between the second derivative of the effective average action and $W^{(2)}_k$ reads
\begin{equation} \label{eq:quantum wave operator}
(P_0 + \U'') W^{(2)}_k = - 1  \ .
\end{equation}

In the following, we will assume that, despite quantum corrections, the quantum wave operator $P_0 + \U''$ remains \textit{Green hyperbolic}; that is, it admits unique advanced and retarded propagators such that
\[
(P_0 + \U'') \Delta^U_{A,R}(f) = f \   \text{and} \ \supp\Delta^U_{A}(f)  \subset J^-(\supp f) \ ,
\supp\Delta^U_{R}(f)  \subset J^+(\supp f).
\]

There is a standard procedure to intertwine the free and the quantum wave operators $P_0$ and $P_0 + \U''$, see \cite{Drago2015}. In fact, consider the operator $(1 - \Delta_{R}^U \U'')$ applied to any function $f$; we have
\[
(P_0 + \U'') (1 - \Delta_R^U \U'') f = P_0 f.
\]
It follows that the operators $(1- \Delta_{A,R}^U \U'')$ intertwine between the free and quantum wave operators. We call $(1- \U'' \Delta_{A}^U)$ and 
$(1- \Delta_{R}^U \U'')$ the \textit{advanced/retarded M\o ller operators}.

We can now rewrite the M\o ller operators in terms of the propagators for the free theory and the effective potential. We start from the defining property of $\Delta_{A,R}^U$, that they are fundamental solutions of the QEOM:
\[
(P_0 + \U'') \Delta_{A,R}^U = 1 \ .
\]
It follows that
\[
P_0 (1 + \Delta_{A,R} \U'') \Delta_{A,R}^U = 1 \ ,
\]
and so
%\color{blue}
the following recursive formula for $\Delta_{A,R}^U$ holds
\begin{equation}\label{eq:DeltaURA}
\Delta^U_{A,R} = \Delta_{A,R} (1-  \U'' \Delta^U_{A,R}).
\end{equation}
The last relation can be used to obtain a series representation of $\Delta^U_{A,R}$ in terms of powers of $\U''$.
%\color{black}
%
%
%\[
%\Delta_{A,R} \U'' = (1+ \Delta_{A,R} \U'')^{-1} \Delta_{A,R} \ .
%\]
%The above formula holds only if the operator $(1 + \Delta_{A,R} \U'')$ admits an inverse.
%From this expression, we can rewrite the M\o ller operators using the Neumann series for the operator inverse, so that
%\begin{multline} \label{eq:series for the propagators}
%1 - \Delta_A^U \U'' = 1 - \sum_{n=0} (-\Delta_{A,R} \U'')^n \Delta_{A,R} \U'' = 1 - \sum_{n=1} (-\Delta_{A,R} \U'')^n \\
%= \sum_{n=0} (-\Delta_{A,R} \U'')^n = (1+ \Delta_{A,R} \U'')^{-1} \ .
%\end{multline}
%
%Notice, however, that while the operator $1-\Delta_{A,R}^U \U''$ is always well-defined, the above expressions for the M\o ller operators are well-defined only when the perturbative series converge.
%
Thanks to the M\o ller operators, we can write solutions and propagators of the quantum wave operator $P_0 + \U''$ in terms of the solutions and propagators of the free theory and of the effective potential $\U''$.
%
%
%In what follows, we want to compute the propagator for the interacting theory in terms of the propagators of the free theory and the effective potential $U_k$.  We recall that ${\color{blue} -i}W^{(2)}_k = \langle \chi(x) \cdot_T \chi(y) \rangle_c$ is the Feynman propagator for the interacting theory. We denote it $G_k(x,y)$ to highlight the fact that we compute it as the propagator for the QEOM coming from the effective average action, and not as the second derivative of the generating functional $W_k$. 

In what follows, we compute the regularised propagator $:G_k:_{\tilde H_F}$ for the interacting theory in terms of the propagators of the free theory and the effective potential $U_k$.  
We recall in particular that the regularised propagator $:G_k:_{\tilde H_F}$ is a solution for the quantum wave operator
$( \Gamma_k^{(2)} - q_k) :G_k:_{\tilde H_F} =0$ up to known smooth terms.
%\color{blue}
Furthermore, denoting by $\mathcal{O}\subset \mathcal{M}$ the support of $V $ and of $q_k$, which is a compact set because of the cut-off functions used in their construction, it holds by causality that 
\[
 :G_k(x,y): = \Delta_F(x,y) - H_F(x,y) = w(x,y) \ , \qquad \forall  x,y \in\mathcal{M}\setminus J^+(\mathcal{O}).  
\]
%\color{black}
%
%
% We want to express  $:G_k:_{\tilde H_F}$ in terms of the effective potential $U_k$, in order to write the RG flow equations as a closed differential equation for $U_k$. To this end 
%
%First of all, we notice that, in the free case, $V=0$, $U^{(2)}_k = -q_k$ and this implies that in the limit where also $k\to0$ we have
%\[
% :G_k: \to_{V=0;q_k=0}  :G_k^{(0)}: = \Delta_F(x,y) - H_F(x,y) = w(x,y) \ ,
%\]
%%\textbf{\color{red} FALSE! It needs to be a k-dependent smooth function w, otherwise there is no flow in Minkowski.}
%where we recall that $w$ is the smooth part of the state two-point function, $\Delta_+$ is the two-point function and $H$ is the Hadamard parametrix. It follows from the above that $\tilde H_F \to_{V=0;q_k=0} H_F$ and $G_k \to_{V=0;q_k=0} \Delta_F$.
%
%In the above equation, we are assuming that the state is quasi-free for the free theory, $\omega(\chi) = 0$. Notice that this does not imply that the state is quasifree for the interacting theory, so that in general $\phi(x) = \langle \chi(x) \rangle \neq 0$.
%
%Now, the normal-ordering $:G_k:_{\tilde H_F}$ satisfies the QEOM
%\[
%(P_0 + U^{(2)}_k) :G_k:_{\tilde H_F} = 0 \ .
%\]
Since $:G_k:$ is a bi-solution of the QEOM, and it reduces to $w$ in %\color{blue} 
the past of the supports of $V, q_k$, and $j$,
%\color{black}
%the free limit it reduces to $w$, 
we arrive at
\begin{equation} \label{eq:G_k double series}
:G_k: = (1 - \Delta_R^U U^{(2)}_k ) w (1 - U^{(2)}_k \Delta_A^U  )\ . 
%=  \sum_{n=0,m=0} (-\Delta_R U^{(2)}_k)^n w (-\Delta_A U^{(2)}_k)^n \ .
\end{equation}
Finally, thanks to the above expression, we conclude that the RG flow equations can be rewritten as differential equations for the effective potential $U$ as
\begin{equation} \label{eq:RG for U}
\partial_k U_k = - \frac{1}{2} \int_{\mathcal{M}}  \partial_k q_k(x) (1 -\Delta_R^U \U'') w (1 - \U'' \Delta_A^U ) \dd\mu_x \ .
\end{equation}
%
%
%
%Despite these issues, the decomposition \eqref{eq:decomposition effective average action} can still be used to rewrite $G_k$ as a Dyson-type series,
%\begin{equation} \label{eq:series for G_k}
%G_k(x,y) = \sum_{n \geq 0} (- \Delta_{F,k} U^{(2)}_k)^n \Delta_{F,k} \ .
%\end{equation}
%The choice of the Green function $\Delta_{F,k}$ comes from the free case. In fact, when $V=0$, we have 
%\[
%G_k(x,y) = \omega(S(Q_k) \cdot_T \chi(x) \cdot_T \chi(y)) = \Delta_{F,k}(x,y) \ ,
%\]
%which has been proven in \cite{DDPR2022}. This uniquely fixes the starting element of the series in \eqref{eq:series for G_k}, since $U^{(2)}_k = \mathcal O(V)$.
%
%With the series \eqref{eq:series for G_k}, the renormalization flow equation can be written as an equation for the effective potential
%\begin{equation}
%\partial_k U_k = \frac{i}{2} \lim_{y \to x } \int_x \partial_k q_k(x) \left [ \sum_{n \geq 0}  (- \Delta_{F,k} U^{(2)}_k)^n \Delta_{F,k} - \tilde H_F(x,y)  \right ] \ .
%\end{equation}
%In the next section, we show how to compute $\tilde H_F$ from the effective potential, in order to obtain a closed equation for $U_k$.
%\color{blue}
In the next section we discuss existence and uniqueness of local solutions of this equation.
%\color{black}

\section{Existence of local solutions}\label{sec:existence of local solutions}
\subsection{Local Potential Approximation}

In this section, we would like to prove an existence theorem for local solutions of the RG flow equations. In order to do so, we restrict our attention to the \textit{Local Potential Approximation}: in this approximation, the effective potential is a local functional which does not contain derivatives of the fields. Furthermore, we consider the case in which the classical field $\phi$ is constant throughout spacetime. 

More precisely, the \textit{Local Potential Approximation} (LPA) assumes that the effective potential and its second functional derivative are 
\begin{equation} \label{eq:LPA}
U_k(\phi) = \int_{\mathcal{M}} u(\phi(x), k) f(x) \dd\mu_x  \ ,
\quad 
U_k^{(2)}(\phi)(x,y) = \partial_\phi^2 u(\phi(x), k) f(x) \delta(x,y)\ ,
\end{equation} 
where $f\in C^{\infty}_0(\mathcal{O})$ is an adiabatic cutoff ($f\geq 1$ and $f=1$ on the relevant part of the spacetime we are working with), which is inserted to keep the theory infrared finite, and  $\mathcal O \in \mathcal M$ is a compact region in the space-time containing the support of $f$.
Notice that what we call here LPA slightly differs from the usual approximation found in the physics literature. In fact, it is standard practice to expand the effective potential around an arbitrary background $\phi = \bar \phi + \varphi$, and then retaining in the effective potential only terms that are quadratic in the fluctuation field $\varphi$. This greatly simplifies the structure of the quantum wave operator, which in this way is approximated by an operator that is linear in the fluctuation field. On the contrary, even though we assume that the effective potential does not contain derivatives of the field, we are retaining its full non-linear dependence on the field $\phi$, without expanding on a fixed background.
%{\color{red} We should assume for simplicity that $\mathcal{M}$ is the Minkowski spacetime so that we have a simple expression for $\Delta_R$.}

We further recall that the background spacetime $\mathcal M$ is ultra-static. This assumption simplifies the explicit form of the retarded and advanced propagators for the free theory $\Delta_{A,R}$, and it allows for simple estimates of their norms. However, these estimates can be easily generalised to static spacetimes, and are known to holds in some special cases, such as de Sitter space.

Finally, in the simplest approximation, we choose the field $\phi$ to be a constant throughout the space-time, so that $u(\phi, k)$ and $ \partial_\phi^2 u(\phi, k)$ are constant in space. The arbitrary function $u(\phi,k)$ and its second field derivative $\partial^2_\phi u$ thus determine the effective potential, and so the effective average action.

In the limit where $V\to 0$ the effective potential reduces to $Q_k$ and $u$ reduces to $-q_k\phi^2/2$. We shall take this into account in fixing the initial conditions for $u$.

Thanks to this approximation, the second derivative of the effective potential $\U''$ appearing in the QEOM reduces to a perturbation of the free wave operator $P_0$ with a smooth external potential that has compact support, and in the limit where $f\to1$ on $\mathcal{M}$ the potential reduces to a mass perturbation.
It follows that many techniques of the generalised principle of perturbative agreement \cite{Drago2015} become readily available.
% , since the difference between the free and quantum wave operators is simply a mass term.  

In particular, it is known that the interacting advanced and retarded propagators $\Delta_{A,R}^U$ are given by the free propagators $\Delta_{A,R}$ associated to $P_0$, with a mass modified by the external potential. In particular the recursive relations given in \eqref{eq:DeltaURA} permits to analyse analytically how $G_k$ depends on $u$. 
%Therefore, the series for the interacting propagator $G_k$ \eqref{eq:G_k double series} converges as well.}
%In particular, it is known that the perturbative series for $\Delta_{A,R}^U$ in terms of the free propagators $\Delta_{A,R}$ converge to the same propagators {\color{blue} associated with $P_0$ perturbed with  external potential, and thus}
% the series for the interacting propagator $G_k$ \eqref{eq:G_k double series} converge as well. 

%\color{blue}
By the LPA, the RG flow equation \eqref{eq:RG for U} becomes a partial differential equation for $u(\phi,k)$.
Thus, we are interested in studying the existence and uniqueness of solutions of the problem associated with the RG flow equation \eqref{eq:RG for U}, supplemented with suitable boundary conditions and a set of initial values explicitly given in terms of the functions $\psi$ and $\beta$ as:
\begin{equation}\label{eq:RG for u}
\begin{cases} 
\partial_k u = G_k(\partial_\phi^2u) \ , \\
u(\phi, a) = \psi \ , \\
\left.u\right|_{\partial X \times [a,b]} = \beta \ .
\end{cases}
\end{equation}
where the function $G_k$ is defined as 
%%%\begin{equation}
%%%G(k,\partial_\phi^2u) :=- \frac{1}{2}  \lim_{y\to x} \partial_k q_k(x) (1 -\Delta_R^u f \partial_\phi^2 u) w (1 - \Delta_A^u f \partial_\phi^2 u)(x,y)\ .
%%%\end{equation}
%\begin{equation}\label{eq:G}
%G_k(\partial_\phi^2u) :=- \frac{1}{2 \lonenorm{f}}  \int_{\mathcal{M}} \dd\mu_x \partial_k q_k(x)  \lim_{y\to x} 
%%\left\{ w(x,y)-  (\partial_\phi^2 u \Delta_R^u f \otimes 1)(w) (x,y) - (1\otimes 
%%\partial_\phi^2 u \Delta_A^u f)(w)(x,y) + (\partial_\phi^2 u \Delta_R^u f \otimes \partial_\phi^2 u \Delta_A^u f)(w)(x,y)
%%\right\}
%\left\{(1-\partial_\phi^2 u \Delta_R^u f) \otimes (1-\partial_\phi^2 u \Delta_R^u f)(w)(x,y)
%\right\}
%\ .
%\end{equation}
\begin{equation}\label{eq:G}
G_k(\partial_\phi^2u) :=- \frac{1}{2 \lonenorm{f}}  \int_{\mathcal{M}} \dd\mu_x \partial_k q_k(x)  
%\lim_{y\to x} 
%\left\{ w(x,y)-  (\partial_\phi^2 u \Delta_R^u f \otimes 1)(w) (x,y) - (1\otimes 
%\partial_\phi^2 u \Delta_A^u f)(w)(x,y) + (\partial_\phi^2 u \Delta_R^u f \otimes \partial_\phi^2 u \Delta_A^u f)(w)(x,y)
%\right\}
\left\{(1-\partial_\phi^2 u \Delta_R^u f) \otimes (1-\partial_\phi^2 u \Delta_R^u f)(w)(x,x)
\right\}
\ .
\end{equation}$w \in C^{\infty}({\mathcal M}^2) $ is a given symmetric smooth function (the smooth part of the chosen background state);
$f\in C^\infty_0(\mathcal{M})$ is the positive cutoff function used in $U$ and
 $\lonenorm{f}$ is the $L^1$ norm of $f$ computed with respect to the standard measure on $\mathcal{M}$;
%the  is a point in the domain of the cutoff $f$ where $f=1$, 
$q_k$ is the integral kernel of the adiabatic regulator $Q_k$, which is assumed to be smooth and with compact support in $x$.
%For simplicity we chose 
%\[
%q_k(x) = k^2 f(x)^2
%\] 
%where $f$ is the same cutoff used in $U$, but the construction we are going to present can be generalized to more complicated functions of $x$ and $k$.
$\Delta_{R}^u:C^{\infty}_0(\mathcal{M})\to C^\infty(\mathcal{M})$ is the retarded fundamental solutions of $(P_0+ f \partial_\phi^2 u) g=0$ which coincides with $\Delta_R^U$ used in other part of the paper and it exists and it is unique because $P_0+ f \partial_\phi^2 u$ is a Green-hyperbolic operator \cite{Baer2015}.
Furthermore, in the integrand in \eqref{eq:G} $f$ is seen as a multiplicative operator which maps $C^{\infty}(\mathcal{M})\to C^{\infty}_0(\mathcal{M})$ and $1$ is the identity map in $C^{\infty}(\mathcal{M})$.
Notice that $\partial_\phi^2 u$ is constant with respect to $P_0$.   
Furthermore, thanks to the support properties of $f$
we have that $O:=(1-\partial_\phi^2 u \Delta_R^u f) \otimes (1-\partial_\phi^2 u \Delta_R^u f)$ is a linear operator on $C^{\infty}(\mathcal{M}\times \mathcal{M})$ to itself. 
Since $w$ is smooth on $\mathcal{M}$ the evaluation of $Ow$ on $(x,x)$ can be easily taken and the integral over $\mathcal{M}$ is finite because $q_k$ is of compact support.
%, namely the limit $y\to x$ of $Ow(x,y)$ can be easily taken and coincides with the evaluation of $Ow$ on $(x,x)\in \mathcal{M}^2$.

%{\color{red} CONTROLLARE IL SEGNO DAVANTI A 1/2 !!!!
To keep the analysis of this part as simple as possible, we shall assume
\begin{equation}\label{eq:qk}
q_k(x):=  (k_0+\epsilon k) f(x) 
\end{equation}
where $f$ is the same spacetime cutoff used in $U$ and where $k$ is assumed to have the dimension of a mass squared. With this choice, $\partial_k{q_k} = f(x)$ and it is independent on $k$. 
We furthermore observe that the contribution proportional to $k_0$ is constant in $k$ and it can always be reabsorbed in a redefinition of the mass of the free theory.
Many other choices, like the more usual $q_k(x) = k^2 f(x)$ can be brought to the same case using $k^2$ in the equation in place of $k$. 
%
%ERRORE nelle stime dei due lemmi finali dove serve che sia piccolo $D\sigma$ bisogna forse toglierlo da $U$ ????
%
%}
%
%\color{black}

The function $u$ in \eqref{eq:G} is a smooth function on compact spaces, and therefore
the tame Fr\'echet space we are working with is $F=C^{\infty}(X\times [a,b])$, 
%\color{blue}
where $X$ is a compact space in $\mathbb{R}$ which contains all possible values of $\phi$ and $k$ is in the positive interval $[a,b] \subset \mathbb{R}^+$, because the sign of $k$ is always assumed to be positive.
\color{black}
This space is Fr\'echet with seminorms
%\color{blue}
\begin{equation}\label{eq:seminorm}
\norm{u}_n = \sum_j^n \sum_{|\alpha| = j }\sup_{\phi, k} \abs{D^\alpha u(\phi, k) } \ ,
\end{equation}
where $\alpha\in \mathbb{N}\times \mathbb{N}$  is a multi-index, and thus the derivatives $D^\alpha$ are taken both in $\phi$ and $k$.
%\color{black}
The space $F$ is tame because it is the space of smooth functions over a compact space \cite{Hamilton1982}. 

%\color{blue}
To have uniqueness of the solution of \eqref{eq:RG for u} we need to provide suitable boundary conditions and to prescribe initial values. We thus assume that 
\begin{equation}\label{eq:initial-boundary-condition}
u(\phi,a)=\psi, \qquad \left.u\right|_{\partial X \times [a,b]} = \beta
\end{equation}
where $\psi$ is a given smooth function on $X$ and $\beta$ is a given smooth function on $\partial X \times [a,b]$ compatible with $\psi$. 
To impose the initial values and the boundary conditions we introduce the tame Fr\'echet subspace of $F$  
\[
F_0 := \{ u\in F \,|\, u(\phi,a)=0,\, \left.u\right|_{\partial X \times [a,b]} = 0 \}.
\]
The solution $\tilde{u}$ of \eqref{eq:RG for u} we are looking for is then of the form 
\begin{equation}\label{eq:ub}
\tilde{u} = u_b+u ,\qquad u\in F_0
\end{equation}
where $u_b$ is a given element of $F$ selected in such a way that it satisfies the boundary conditions and respects the initial values given in \eqref{eq:initial-boundary-condition}. 

We also further assume that $\partial^2_\phi u$ and its second derivatives lie in a suitably small neighbourhood of $0$, that is, $\norm{u}_4 \leq A$ for some positive constant $A$.
%\color{black}

To prove existence of local solutions of the RG flow equations, we make use of Nash-Moser theorem in Hamilton's formulation. To do so, we need to prove the validity of the strong assumptions of the theorem. We already remarked that $u$ lives in a suitable tame Fréchet space  $F_0$.
The RG flow equations, acting on $u$, determine a RG operator $\mathcal{RG}: \mathcal U \subset  F_0 \to F$, given below in Definition \ref{def:RG}. To use the Nash-Moser theorem, we further need to prove that i) the RG operator acting on $u$, is a tame smooth operator; ii) that its linearisation $D\mathcal{RG}(u): F_0 \to F$ is tame smooth as well; iii) that the linearisation of the RG operator admits a  unique inverse $D\mathcal{RG}^{-1}(u):F\to F_0 $  for every $u\in U$, and that the inverse is tame smooth. We will prove each of these assumptions in the following propositions. Since we can prove the assumptions of Nash-Moser theorem, it follows that the RG operator admits a local inverse.
The solution of the RG flow equations is then determined 
%by
%\[
%u = \mathcal{RG}^{-1}(0) \ .
%\]
%\color{blue}
as the unique solution of the equation \cite{Hamilton1982}
%\color{red}
\begin{equation}\label{eq:first-order-equation-recursion}
\frac{d}{dt} u_t = - c D\mathcal{RG}^{-1}(S_t u_t) S_t \left( \mathcal{RG}(u_t) \right)
\end{equation}
%\color{blue}
%\[
%\frac{d}{dt} u_t =  - c D\mathcal{RG}^{-1}(S_t u_t) S_t \mathcal{RG}(u_t)
%\]
with a given $u_0=0$. In this equation $c$ is a positive fixed arbitrary constant and $S_t$ is a smoothing operator, see e.g.  \cite{Hamilton1982}.

If $\mathcal{RG}$ is a smooth tame map, if $D\mathcal{RG}(u)$ admits an unique inverse for every $u$ in a suitable subset of $F_0$, and if the inverse $D\mathcal{RG}^{-1}$ is also tame, a unique solution of equation \eqref{eq:first-order-equation-recursion} exists for all $t$ such that the limit of the sequence of approximated solutions converges to a solution of the RG flow equations, $\lim_{t\to\infty} u_t = u_\infty$ is such that $\mathcal{RG}(u_\infty)=0$, as it is proved in \cite{Hamilton1982}.

% Notice that we use here the fact that 
% $D\mathcal{RG} (U\subset C^\infty )\times C^{\infty} \to  C^{\infty}$ admits an unique inverse for every $f$ in $U$

%\color{red}
%{\tt  alternativa}
%%
%by an iteration procedure where
%\[
%u_{n+1} = u_n +   D\mathcal{RG}^{-1}(u_n) S_n \mathcal{RG}(u_n)
%\]
%and where $S_t$ is a suitable smoothing operator \cite{Hamilton1982}.
\color{black}

\subsection{The RG operator is tame smooth}

% \color{red}

% Implement boundary conditions and initial values restricting elements of $F$ to $F_0$.
% \[
% F_0:=\{u\in F \, |\, u(0,\phi)=\psi(\phi), \, u |_{\partial X \times [a,b]} = \beta\}.
% \]

%\color{blue}
Following the strategy presented in the last section, we start with a formal definition.
\begin{definition}[RG operator]\label{def:RG}
Let $u_b\in F$ be such that it satisfies the initial values and the boundary conditions given in \eqref{eq:initial-boundary-condition}.
The \textit{RG operator} $\mathcal{RG} : \mathcal U \subset F_0 \to F$ is defined as
\begin{equation} \label{def: RG operator}
\mathcal{RG} : u \mapsto \mathcal{RG}(u) := 
\partial_k (u+u_b) - G_k(\partial_\phi^2 (u+u_b))
%
%\\
%\partial_k u + \frac{1}{2} \int_{\mathcal M} \dd x \dot q_k(x) (1 - \Delta_{R}(f) \partial_\phi^2 u)^{-1}(x,z_1) w(z_1,z_2) (1 - \Delta_{A}(f) \partial_\phi^2 u)^{-1}(z_2, x) \ ,
\end{equation}
where $G_k$ is given in \eqref{eq:G}. 
%$w \in C^\infty_0(\mathcal M)$, $\dot q_k := \partial_k q_k$. $\Delta_{R,A}(f)$ is the action of the retarded and advanced propagators on the IR cut-off function $f$.
\end{definition}
%\color{black} 
%{\color{red} (CONTROLLARE I SEGNI E I COEFFICIENTI)}

%In the LPA, the contribution from the effective potential to the r.h.s. of \eqref{def: RG operator} is a mass perturbation. In \cite{Drago2015}, it has been proven that mass perturbations of the free Feynman propagator converges to a Feynman propagator with a modified mass, so that we can write the RG operator as
%\[
%\mathcal{RG}(u) = \partial_k u - \frac{1}{2}  \Delta_{F,k,u}(x,y)(\dot q_k, w) \ ,
%\]
%where $\Delta_{F,k,u}$ is the Feynman propagator with mass $k^2 + k^2 \partial_\phi^2 u$.

As a first step in the proof of existence of local solutions, we want to prove that the RG operator is of the right class to apply the Nash-Moser theorem, i.e., it is tame smooth. In order to prove it,  
%\color{blue}
we start considering $G_k$ in \eqref{eq:G}.
We observe that since with $q_k$ given in \eqref{eq:qk}  $\partial_k q_k$ is constant in $(\phi,k)$, we have that 
${G}_k$ depends on $(\phi,k)$ only through  $\partial_\phi^2 \tilde{u}$, where we recall that $\tilde u = u + u_b$.
Consider now $G_k$, written as
\[
G_k(\partial_{\phi}^2 \tilde{u}) = -\frac{1}{2\lonenorm{f}} \int_{\mathcal{M}} \dd\mu_x   \partial_k q_k(x)    :{G}_k:(x,x)
\]
for $\tilde{u}\in F$.
We analyse how $:{G}_k:(x,x)$ depends on $\partial_\phi^2 \tilde{u}$. 
Notice that $:{G}_k:(x,y)$ can be given explicitly as
\[
:{G}_k:(x,y) := \int \dd\mu_{z_1} \dd\mu_{z_2}  (\delta - \Delta^U_{R} \U'')(x,z_1) w(z_1,z_2) (\delta - \U'' \Delta^U_{A} )(z_2, y) \ 
\]
where $\delta$ is the Dirac delta function (the integral kernel of the identity). 
Recalling that $(1 - \Delta^U_{R} \U'')\circ(1 + \Delta_R \U'') =1$, using the recursive relations given in \eqref{eq:DeltaURA},
%
%
%we can rewrite the first inverse in terms of its Neumann series
%\[
%:{G}_k:(x,y) = \int_{z_1, z_2} \sum_{n=0} (-\Delta_R \U'')^n w(z_1, z_2) (1 - \U'' \Delta^U_{A} ) \ .
%\]
%The last equation 
we obtain a 
%can be re-expressed $:{G}_k:(x,y)$ as a 
recursive formula for $:{G}_k:(x,y)$:
\begin{equation} \label{eq:recursive formula G}
:{G}_k:(x,y) = \tilde{w}(x,y) - \int \dd\mu_z \Delta_R \U''(x,z) :{G}_k:(z,y) \ ,
\end{equation}
%As a first step in the proof of existence of local solutions, we want to prove that the RG operator is of the right class to apply the Nash-Moser theorem, i.e., it is tame smooth. In order to prove it, we start considering $:G_k:$, which explicitly is given by
%\[
%:G_k: = \int_{z_1, z_2} (1 + \Delta_{R} \U'')^{-1}(x,z_1) w(z_1,z_2) (1 + \Delta_{A} \U'')^{-1}(z_2, y) \ .
%\]
%Rewriting the first inverse in terms of its Neumann series, we get a recursive relation for $:G_k:$ that can be written as
%\[
%:G_k: = \int_{z_1, z_2} \sum_{n=0} (-\Delta_R \U'')^n w(z_1, z_2) (1 + \Delta_{A} \U'')^{-1} \Rightarrow
%\]
%\begin{equation} \label{eq:recursive formula G}
%:G_k: = w - \int_z \Delta_R \U''(x,z) :G_k:(z,y) \ .
%\end{equation}
where
\[
\tilde{w}(x,y) := \int \dd\mu_z w(x,z) (\delta-\U''\Delta_A^U )(z,y).
\]
This recursive relation can be used to get estimates of $:G_k(x,y):$, for $x,y$ contained in some compact region of the spacetime $\mathcal{M}$. 
First of all, we can prove the following Lemma, on estimates of the retarded propagator $\Delta_R^U g$ acting on some compactly supported smooth function $g$.

%In fact, we can prove the following Lemma. 
%
% 
%
%
%and taking the absolute value, the recursive relation gives a first inequality satisfied by the inverse $:G_k:$
%\begin{equation} \label{eq:triangle inequality G}
%\abs{:G_k:(x,y)} \leq \abs{ \tilde{w}(x,y)} + \int_z \abs{\Delta_R \U''(x,z)} \abs{:G_k(z,y):} \ .
%\end{equation}
%
%With similar steps, we can get an estimate for $\tilde{w}$ in its second variable as
%\begin{equation}\label{eq:def tildew}
%\abs{\tilde{w}(x,y)} \leq \abs{ w(x,y)} + \int_z  \abs{\tilde{w}(x,z)} \abs{\U'' \Delta_A (z,y)} \ .
%\end{equation}
%
%\color{black}
%Starting from the above two inequalities, we can prove that the RG operator is tame smooth. 
%The proof is based on estimates of $:G_k:$ which relies on Gr\"onwall inequality in integrated form.

%\color{blue}
\begin{lemma} \label{eq:lemma-DeltaUf}
Let $\mathcal{M}$ be a 
ultra-static spacetime and let $t$ be a time function. Let $\tilde{u}\in F$, and consider 
\[
h= \Delta^U_R g 
\]
where $g$ is a compactly supported smooth function on $\mathcal{M}$. 
It then holds that $h$ is a past-compact smooth function with compact support on every Cauchy surface $\Sigma$. Moreover, recalling \eqref{eq:DeltaURA}, writing $h$ as
\[
h= (1-\Delta^U_R U^{(2)})  \varphi \ ,
\]
where $\varphi = \Delta_R g$, the following estimates hold: 
%We want to prove that  
\begin{equation}\label{eq:ineq1}
\lnorm{h}^{t}_\infty  \leq  c \ltwotwonormt{h}{t}  \leq 
c \ltwotwonormt{\varphi}{t}   
%e^{\|\tilde{u}\|_2 \int_{-\infty}^t \dd \tau  (t-\tau) \lnorm{f}^\tau_{\infty} }
e^{C |\partial_{\phi}^2 \tilde{u}|  }
%\int_{-\infty}^t \dd \tau  (t-\tau) \lnorm{f}^\tau_{\infty} }
\leq
c
e^{C \norm{\tilde{u}}_2  }
 \ltwotwonormt{\varphi}{t}   
\end{equation}
and
\begin{equation}\label{eq:ineq2}
\lnorm{h}^{t}_\infty  \leq  c \ltwotwonormt{h}{t}  \leq 
c
e^{C \norm{\tilde{u}}_2 }
 \int_{-\infty}^t \dd \tau  
(t-\tau) \ltwotwonormt{g}{\tau}  
\leq 
\tilde{C}e^{C \norm{\tilde{u}}_2 } \sup_{\tau\leq t} \ltwotwonormt{g}{\tau} \ .
\end{equation}
In the above inequalities, $C>0$ is a positive constant, which depends on the support of $f$ in $U$ but not on $\tilde{u}$; similarly, 
$\tilde{C}>0$ depends only on the support of $g$ and $c$ is positive and does not depend on $U$.
Furthermore,  $\ltwotwonormt{\cdot}{t}$ is the norm on the Sobolev space $W_{2,2}(\Sigma_t)$ and $\lnorm{\cdot}_\alpha^{t}$ is the norm on $L^{\alpha}(\Sigma_t)$ where $\Sigma_t=\{x \in \mathcal{M}| t(x)=t\}$ is the Cauchy surface at fixed time $t$.
\end{lemma}
\begin{proof}
We recall that both $\Delta_R$ and  $\Delta^U_R$ map past-compact smooth functions to past-compact smooth functions, hence both $\varphi = \Delta_R g$ and $h=\Delta^U_R$ are smooth and past-compact.
We also recall that 
\[
\Delta^U_R 
=  \Delta _R  (1-\U'' \Delta^U_R ) 
=  (1-\Delta^U_R \U'')\Delta _R \ .
\]
Since $\U''= f \partial_\phi^2 \tilde{u} $, where $f$ is a smooth compactly supported function and $\partial_\phi^2 \tilde{u}$ is constant on $\mathcal{M}$, the following recursive relation holds
\[
h = \Delta_R g - \Delta_R \U''  h = \varphi - \Delta_R \U''  h.
\]
Now, let $D$ be the (positive) Laplace operator on $\Sigma_t$ constructed with the induced metric on $\Sigma_t$, and define $\omega =\sqrt{D+m^2}$ as the square root of the positive operator $D+m^2$. Hence
\[
h(t,\mathbf{x}) = \varphi (t,\mathbf{x}) - \partial_\phi^2 \tilde{u} \int_{-\infty}^{t} \dd \tau  \frac{ \sin{(\omega(t-\tau))} }{\omega} (f h)  (\tau,\mathbf{x} ) \ .
\]
We thus have   
\[
\lnorm{h}^t_{2} \leq  \lnorm{\varphi}^t_{2} +  |\partial_\phi^2 \tilde{u} |  \int_{-\infty}^{t} \dd \tau  (t-\tau)  \|f\|^\tau_{\infty}  \lnorm{h}^\tau_2  \ ,
\]
or, passing to the Sobolev norm $\lnorm{h}_{2,2}^t= \lnorm{h}_2^t+ \lnorm{D h}^t_2$, we have
\begin{align*}
\lnorm{h}^t_{2,2} & \leq  \lnorm{\varphi}^t_{2,2} +  |\partial_\phi^2 \tilde{u} |  \int_{-\infty}^{t} \dd \tau  (t-\tau)  \lnorm{f h}^\tau_{2,2}  
\\ 
 & \leq  \lnorm{\varphi}^t_{2,2} +  |\partial_\phi^2 \tilde{u} |  \int_{-\infty}^{t} \dd \tau  (t-\tau) (\|f\|^\tau_{\infty} + \sup_{i} 2 \| \partial_i f\|^\tau_{\infty}  + \|D f\|^\tau_{\infty})  \lnorm{h}^\tau_{2,2}
 \\
 & \leq  \lnorm{\varphi}^t_{2,2} + C |\partial_\phi^2 \tilde{u} |   \int _{a}^t   \dd \tau \lnorm{h}^\tau_{2,2}
\end{align*}
where $a=\inf_{x\in\supp f}\{ t(x)  \}$ and
for a suitable positive constant $C$ independent on $\partial_\phi^2 \tilde{u}$. $C$ is in fact finite because $f$ is smooth and with compact support on $\mathcal{M}$.

Applying the Gr\"onwall Lemma in integrated form to the previous inequality we obtain
\begin{align*}
\lnorm{h}^t_{2,2} \leq  \lnorm{\varphi}^t_{2,2} e^{C |\partial_\phi^2 \tilde{u}| }.
\end{align*}
To conclude the proof of the first inequality \ref{eq:ineq1}, we observe that $\Sigma_t$ is a three dimensional space, and so by standard arguments we have 
\[
\lnorm{h}^t_\infty 
\leq 
\lnorm{\hat{h}}^t_1  \leq \lnorm{(1+D){h}}_2  \lnorm{(1+D)^{-1}}_2 \leq 
c  \lnorm{{h}}_{2,2}.
\] 
where the $\lnorm{(1+D)^{-1}}_2$ is the $L$-$2$ norm of $(1+{D})^{-1}$.  
To prove \eqref{eq:ineq2} we use  \eqref{eq:ineq1} for $\varphi = \Delta_R g$.
Recalling that
\[
\varphi (t,\mathbf{x}) =  \partial_\phi^2 \tilde{u} \int_{-\infty}^{t} \dd \tau  \frac{ \sin{(\omega(t-\tau))} }{\omega} (g)  (\tau,\mathbf{x} ).
\]
and taking the Sobolev norms we have 
\[
\lnorm{\Delta_R g}_{2,2}^t \leq \int_{-\infty}^t \dd \tau  (t-\tau )   \lnorm{ g}_{2,2}^{\tau}.
\]
\end{proof}

Starting from the above analysis and the previous Lemma, we can prove that the RG operator is tame smooth.

%\color{black}

\begin{proposition} \label{prop: RG tame smooth}
%\color{blue}
Assume that $\mathcal{U}\subset F_0$ is a small neighbourhood of $0$ so that for $u\in \mathcal{U}$,
$\norm{ u }_2 < A$ for some constant $A$. 
%\color{black}
%Assume that $\partial^2_\phi u$ lies in a neighbourhood of $0$, so that $\norm{\partial^2_\phi u }_0 \leq A$ for some constant $A$. 
Then the RG operator is a smooth tame map.
\end{proposition}
\begin{proof}
%\color{blue}
We start considering $\tilde{u}=u_b+u$ for $u\in F_0$ and for a given $u_b$ which satisfies \eqref{eq:initial-boundary-condition} so that $\tilde{u}\in F$ and it satisfies the prescribed initial values and boundary conditions. 
We recall that from \eqref{def: RG operator}
%\begin{equation} \label{def: RG operator}
\[
\mathcal{RG}(u) 
=
\partial_k (u+u_b) - G_k(\partial_\phi^2 (u+u_b))
\]
where $G_k$ is given in \eqref{eq:G}.  To prove that $\mathcal{RG}$ is tame smooth we just need to prove that $G_k$ is tame smooth for $\tilde{u}\in u_b+\mathcal{U}$. We have actually the following lemma

\begin{lemma}\label{lem:G-smooth}
The functional $G_k$ is a smooth function of $\partial_\phi^2\tilde{u}$. Furthermore, it is tame smooth for  $\tilde{u}\in u_b+\mathcal{U}$.
\end{lemma}
\begin{proof}
We observe that  $\partial_k q_k$ is constant on $X\times[a,b]$; hence, recalling the definition of $G_k$ given in \eqref{eq:G}, we have that $G_k(\partial_\phi^2\tilde{u})$ as a function on $X\times[a,b]$ depends on $(\phi,k)$ only through $\partial_\phi^2\tilde{u}$, that is,  $G_k(\partial_\phi^2 \tilde{u})(\phi,k)=G_k(\partial_\phi^2 \tilde{u}(\phi,k))$.
We also observe that $G_k(\partial_\phi^2\tilde{u})$ depends smoothly on $\tilde{u}\in F$.
% actually, if we compute the 
%
%We start analyzing the fun
%
%The functional $:G_k(x,x):(u)$ depends smoothly on $u\in F$ and furthermore
%$:G_k(x,x):(u)(\phi,k)=
%:G_k(x,x):(u(\phi,k))$ namely as a function on $X\times [a,b]$ it is a composition of maps.
Actually, the $n-$th order functional derivative of $\tilde{G}_k(\tilde{u})=G_k(\partial_\phi^2 \tilde{u})$ with respect to $\tilde{u}$ can be explicitly computed and it is well defined for every $n$; in fact, it is given by  
\begin{equation}\label{def:Al}
\begin{gathered}
\tilde{G}_k^{(n)}(v_1,\dots ,v_n) = 
\frac{(-1)^{n+1}}{\lonenorm{f}} n!
\sum_{l=0}^n  
%\sum_{j\in P_l} \sum_{i\in P_{n-l}}
%\prod_{c=1}^l \prod_{\prod_{d=1}^{n-l}}
\int_{\mathcal{M}}\dd\mu_x \partial_kq_k(x)
\cdot \\
\left\{
(\Delta^U_R f)^l \otimes (\Delta^U_R f)^{n-l} \circ 
(1-\partial_\phi^2 \tilde{u} \Delta_R^U f) \otimes (1-\partial_\phi^2 \tilde{u} \Delta_R^U f)(w)(x,x)
\right\}
\prod_{j=1}^n\partial_\phi^2 v_j\\
=: A_n(\tilde{u})\prod_{j=1}^n\partial_\phi^2 v_j \ .
\end{gathered}
\end{equation}
In the last formula, $f$ in $\Delta^U_R f$ is a multiplicative operator, and $A_n(\tilde{u})$ are suitable functionals of $\tilde{u}$.
Notice that both the cutoff functions $f$ and $q_k$ have compact support. $w$ is a smooth function on 
$\mathcal{M}^2$. Hence for every $\tilde{u}\in F$ the integral which defines $A_n$ can always be taken and it gives a finite bounded result.
We thus have that $G_k$ is a smooth function of $\partial_\phi^2 \tilde{u}$.

To prove that $G_k(\partial_\phi^2\tilde{u})$ is also tame, we proceed as follows.
% \color{red} {\tt Questo sotto probabilmente va cancellato.}
We recall that $\|u\|_2< \| u \|_4 < A$, and that $G_k$ depends on $\phi$ and $k$ only through $\partial_\phi^2\tilde u$, because $\partial_k q_k = f$. 
By direct inspection, we have that 
\begin{equation}\label{eq:Gkstima}
\|G_k\|_n <  \|A_0(\tilde{u})\|_0
+
\sum_{p=1}^n
\sum_{l=1}^p  \|A_l(\tilde{u})\|_0 \| (\partial_\phi^2 \tilde{u})^l \|_{p-l} \ .
\end{equation}
To estimate $\| (\partial_\phi^2 \tilde{u})^l \|_{p-l}$, we use Leibniz rule together with an interpolating argument (See Corollary 2.2.2 in \cite{Hamilton1982}), stating that, for every $f,g\in F$,
\[
\| f \|_n\| g\|_m\leq C (\| f\|_{n+m}\| g\|_0+\| f\|_0 \| g\|_{n+m}) \ .
\]
Hence, by Leibniz rule, we have that  
\[
\|\partial_\phi^2\tilde{u}^l\|_{r}\leq C \sum_{R=(r_1,\dots, r_l),  |R|=r} \prod_{i=1}^l  \|\partial_\phi^2\tilde{u}\|_{r_i} \leq C' \|\partial_\phi^2\tilde{u}\|_r\|\partial_\phi^2\tilde{u}\|_0^{l-1} \ .
\]
Using this in \eqref{eq:Gkstima} we get
\begin{align*}
\|G_k\|_n \
&<  
C\left(
\|A_0(\tilde{u})\|_0
+
\sum_{p=1}^n
\sum_{l=1}^p  \|A_l(\tilde{u})\|_0 \| (\partial_\phi^2 \tilde{u})\|_0^{l-1} \|\tilde{u} \|_{p+2}
\right)\\
&<  C\left(
\|A_0(\tilde{u})\|_0
+
\sum_{p=1}^n
\sum_{l=1}^p  \|A_l(\tilde{u})\|_0   \|\tilde{u}\|_2^{l-1} \|\tilde{u} \|_{p+2}
\right)\\
&<  C\left(
1+ \|\tilde{u}\|_{n+2}\right) \ ,
\end{align*}
where in the last step we used the fact that $\|A_l(\tilde{u})\|_0\leq C(1+\|\tilde{u}\|_2)$. This last inequality is proved in the following Lemma \ref{lem:Al}.
\end{proof}

\begin{lemma}\label{lem:Al}
Consider the functionals $A_l(\tilde{u})$ for $\tilde{u}\in F$ given in \eqref{def:Al}. If $\|\tilde{u}\|_2< A$, it holds that  
\[
\|A_l(\tilde{u})\|_0\leq C(1+\|\tilde{u}\|_2).
\]
\end{lemma}
\begin{proof}
To prove this result we observe that both $\partial_k q_k$ and $f$ are smooth compactly supported functions on $\mathcal{M}$. 
The integral present in \eqref{def:Al} is thus taken on a compact region, even if 
$w$ is a smooth function supported in general everywhere on $\mathcal{M}^2$. 
Now, we need to estimate the action of each $\Delta^U_R f $  and of  
$(1-\partial_\phi^2 \tilde{u} \Delta_R^U f)$ by means of Lemma \ref{eq:lemma-DeltaUf}. 

Actually, Lemma \ref{eq:lemma-DeltaUf} implies that if $g$ is a smooth past-compact function, the following estimates hold:
\[
\lnorm{(\Delta^U_R f)^{n} g}^t_{2,2}
\leq 
\sup_{\tau<t}
\lnorm{(\Delta^U_R f)^{n-1} g}^\tau_{2,2} \tilde{C}e^{C \|\tilde{u}\|_2}
\leq
\sup_{\tau<t}
\lnorm{g}^\tau_{2,2} \tilde{C}^n e^{n C \|\tilde{u}\|_2} \ ,
\]
where the constant $\tilde{C}$ depends on $f$.
Similarly, 
\[
\lnorm{(1-\partial_\phi^2 \tilde{u} \Delta_R^U f) g}^t_{2,2}\leq 
ce^{C \|\tilde{u}\|_2} \ .
\]
We now use these estimates in 
\[
a_{l_1,l_2}(x,y):= \left\{
(\Delta^U_R f)^{l_1} \otimes (\Delta^U_R f)^{l_2} \circ 
(1-\partial_\phi^2 \tilde{u} \Delta_R^U f) \otimes (1-\partial_\phi^2 \tilde{u} \Delta_R^U f)(\chi w\chi )(x,y)
\right\} \ ,
\]
for $l_1+l_2 = n$, and where $\theta$ is a smooth compactly supported function which is equal to $1$ in a region which contains the support of $q_k$ and $f$. 
Thanks to this choice, we can replace $w$ in $G_k$ with $\theta w \theta$, getting
%Thanks to this choice, using $\chi w \chi$ in place of $w$ in $G_k$ or $A_l$ nothing changes. To get that 
\[
\sup_{x\in \supp f} | a_{l_1,l_2}(x,x)| \leq 
\sup_{x,y\in \supp f} | a_{l_1,l_2}(x,y)| \leq      \sup_{t_x,t_y \in \supp f} \lnorm{\theta w \theta}^{(t_x,t_y)}_{4,2}
c^2 \tilde{C}^n e^{(n+2) C \|\tilde{u}\|_2}
\]
where $\lnorm{\cdot}^{(t_x,t_y)}_{4,2}$ is the Sobolev norm for functions defined on  $\Sigma_{t_x}\times \Sigma_{t_y}$. 
Using this estimate sufficiently many times in $A_l$, and recalling that $ e^{ C \|\tilde{u}\|_2} \leq C_1 (1+\|\tilde{u} \|_2)$ for a sufficiently large $C_1$ because $\|\tilde{u}\|_2 < A$, we have the thesis. 
\end{proof}

With this results, we can thus conclude the proof recalling that the linear combinations of smooth tame functionals is tame smooth.
\end{proof}
\color{black}

\color{black}

%\begin{proof}
%In \cite{DDPR2022}, it was proven that, on static spacetimes, the Feynman propagator with mass $M = k^2 + k^2 \partial_\phi^2 u$ applied to any two compactly supported, smooth functions $\dot q_k$ and $w$ satisfies the estimate
%\[
%\abs{\Delta_{F,k,u}(x,y)(\dot q_k, w)} \leq \frac{C}{k \sqrt{1 + \partial_\phi^2 u} } \norm{\dot q_k}_2 \norm{w}_2 \ .
%\]
%Since $\Gamma$ is a Legendre transform $\partial^2_\phi u$ is positive, and the square root is greater than $1$, so it follows that
%\[
%\abs{\Delta_{F,k,u}(x,y)(\dot q_k, w)} \leq C' (1 + \partial^2_\phi u) \norm{\dot q_k}_2 \norm{w}_2 \ .
%\]
%Applying the above estimate to the estimate of the $\mathcal{RG}$ operator we can see that
%\[
%\norm{\mathcal{RG}(u)}_n \leq C(1 + \norm{u}_{n+2}) \ .
%\]
%\end{proof}

\subsection{The linearisation of the RG operator is tame smooth}
%Writing
%\[
%(1 + i \Delta_{F,k}(f) k^2 \partial_\phi^2 u )^{-1} = (P_{0,k} + f k^2 \partial_\phi^2 u)^{-1} P_{0,k} \ ,
%\]
%and denoting $P_{0,k} w = \lambda$, the RG operator simplifies into
%
%\begin{equation}
%   \mathcal{RG}(u) := \partial_k u - \frac{i}{2}\int_{\mathcal M} \dd x \partial_k q_k(x) (P_{0,k} + f k^2 \partial_\phi^2 u)^{-1}(x,y) \lambda(y,x) \ .
%\end{equation}

The first derivative of the RG operator  
%\color{blue}
%is given in terms of the linearised RG operator 
%$L(u) v = D\mathcal{RG}(u)v$
%we shall study and invert below.
%
%
%
defines the linearised RG operator 
$L(u) v = D\mathcal{RG}(u)v$ and by direct inspection it
is given by the linear operator
\[
D\mathcal{RG}(u)v 
% \partial_k v 
%\\ - \frac{1}{2} \partial^2_\phi v \int \partial_k q_k \left \{ (1 + \Delta_R \U'')^{-1} \left [ \Delta_R(f)   (1 + \Delta_R \U'')^{-1} w
%+  w (1 + \Delta_A \U'')^{-1} \Delta_A(f) \right ]  (1 + \Delta_A \U'')^{-1}  \right \} \\
=  \partial_k v -  \sigma \partial^2_\phi v \ ,
\]
where 
\begin{equation}\label{eq:sigma}
\begin{aligned}
\sigma(u) := \frac{1}{\lonenorm{f}}
 \int_{\mathcal{M}^2}  & \dd\mu_x\dd\mu_y \,\partial_k q_k(x)  \Delta_R^U(x,y)f(y)
 \\
& \left\{(1-\partial_\phi^2 u \Delta_R^u f) \otimes (1-\partial_\phi^2 u \Delta_R^u f)(w)(y,x)
\right\}.
\end{aligned}
\end{equation}
The function $\sigma$ as a function on $X\times [a,b]$ depends on $\phi$ through $\partial_\phi^2 u$ and on $k$ through $\partial_k q_k$ and $\partial_\phi^2 u$. With the choice of $q_k$ given in \eqref{eq:qk} $\partial_k q_k$ is constant in $k$, and the only way in which $\sigma$ depends on $(\phi,k)$ is through $u$.

\begin{definition}
Let $u_b\in F$ be a function which satisfies the initial values and boundary conditions given in \eqref{eq:initial-boundary-condition}, and let $\mathcal{U}$ be a neighbourhood of $0$ in $F_0$.
The linearised RG operator is defined as the map
\[
L:(u_b+\mathcal{U})\times F_0\to F
\]
\[
L(u)f  := \partial_k g - \sigma(u)\partial_\phi^2 g.     
\]
where $\sigma$ is the map defined in \eqref{eq:sigma}.
%where 
%\[
%\sigma(u) := {\color{red} -\frac{1}{2}} 
%\\
%G_k(\partial_\phi^2u) :=- \frac{1}{2 \lonenorm{f}}  \int_{\mathcal{M}} \dd\mu_x \partial_k q_k(x).  
%\]
\end{definition}

The following Proposition specifies some of the properties of $\sigma$ that will be useful in the analysis of $L(u)$.

\begin{proposition}
The function $\sigma(u)$ is tame smooth.
\end{proposition}
\begin{proof}
The function $\sigma$ is linear in $q_k$ and $q_k$ is a smooth function of $k$: 
actually recalling  \eqref{eq:qk} $q_k= (\epsilon k +k_0) f(x)$, 
and so $\partial_k q_k$ is constant in $(\phi,k)$. Hence, 
$\sigma$ depends on $k$ and on $\phi$ only through $u$.
Furthermore, the $n-$th order functional derivative $\sigma$ with respect to $\partial_\phi^2 u$ is always well-defined because it equals the $n+1$ order functional derivative of $G_k$ with respect to $\partial_\phi^2 u$, and we already proved in Lemma \eqref{lem:G-smooth} that $G_k$ is a smooth function of $\partial_\phi^2u$.
%\begin{gather*}
%G^{(n)}_k(v_1,\dots ,v_n) = \\
%\frac{(-1)^{n+1}}{\lonenorm{f}} 
%\sum_{l=0}^n   l! (n-l)!  
%%\sum_{j\in P_l} \sum_{i\in P_{n-l}}
%%\prod_{c=1}^l \prod_{\prod_{d=1}^{n-l}}
%\int_{\mathcal{M}}\dd\mu_x \partial_kq_k(x)
%\left\{
%(\Delta^u_R f)^j \otimes (\Delta^u_R f)^{n-j} \circ 
%(1-\partial_\phi^2 u \Delta_R^u f) \otimes (1-\partial_\phi^2 u \Delta_R^u f)(w)(x,x)
%\right\}\\
%\prod_{j=1}^n\partial_\phi^2 v_j. 
%\end{gather*}
%Both the cutoff functions $f$ and $q_k$ are of compact support. $w$ is a smooth function on 
%$\mathcal{M}^2$. Hence for every $u\in F$ the integral can be taken and gives a finite result.
Furthermore, $\sigma$ is a smooth function and it is tame with respect to $u$
because it is related to the functional derivative of $G_k$ which is tame smooth as proven in Lemma \ref{lem:G-smooth}.
% {\color{red} Why is it tame?}
Hence $\sigma$ is tame smooth.
\end{proof}
%
% 
%  
%
%Consider
%\[
%A(u) = \langle g,r(h)\rangle = \langle g, (1-\Delta_R^{U} \partial_\phi^2 f) h \rangle
%\]
%where $g\in C^\infty_0 (\mathcal{M})$, $h\in C^\infty(\mathcal{M})$.
%We have that the functional derivative of  $A(u)$
%  
%
%
%descends  from the Lein
%
%
% is thus 
%
%
%
%\color{blue}

The next proposition shows that, by a suitable choice of smooth functions $w$ (or, equivalently, by suitable choices of states), the assumptions that: i) $\sigma$ is larger than some positive constant $c$, and ii) that $\norm{u}_2 \leq A$ is in some small neighbourhood of $0$, hold.
\begin{proposition}\label{prop:select-sigma}
If the boundary conditions given in \eqref{eq:initial-boundary-condition} are such that $\|\beta\|_2+\|\psi\|_2 < \epsilon$ for a sufficiently small $\epsilon$ and if $u_b$ in \eqref{eq:ub} is chosen to be such that $\|u_b\|_2\leq \epsilon$, then
for certain choices of the function $w\in C^{\infty}(\mathcal{M}^2)$, it exists a neighbourhood $\mathcal{U}\subset F_0$ such that for every $u\in \mathcal{U}$,  $\sigma(u_b+u)\geq c>0$ and $\|u\|_2< A=\epsilon$.
\end{proposition}
\begin{proof}
We recall that 
\[
\sigma(0) = \frac{1}{\lonenorm{f}}
 \int_{\mathcal{M}^2} \dd\mu_x\dd\mu_y \,\partial_k q_k(x)  f(y) \Delta_R(x,y)(w)(y,x).
\]
$\sigma(0)$ is linear in $w$ and it cannot be identically $0$ for every $w$, hence it is possible to choose a $w$ such that $\sigma(0)\geq (2\epsilon C + c)>0$, where $C > \sup_{\lambda\in[0,1]} \|\sigma^{(1)}(\lambda(u+u_b))\|_0 $. Moreover, 
$\sigma$ depends smoothly on $u$. We can choose $u_b$ so that  $\| u_b\|_2\leq  (\|\beta\|_2+\|\psi\|_2 ) < \epsilon $ and we can choose a sufficiently small $\mathcal{U} \subset F_0$ such that every $u\in\mathcal{U}$ is such that $\|u\|_2<\epsilon$. 
Hence, the smoothness of $\sigma(u)$ implies that  
%.  ({\tt \color{red} take into account $\|\sigma\|_1<A$}) 
\begin{equation}\label{eq:stima-sigma}
\begin{aligned}
\sigma(u)  & =
\sigma(0) + \int_0^1\dd \lambda \,  \frac{d}{d\lambda}\sigma(\lambda (u+u_b))
\\
& \geq 
\sigma(0) - \sup_{\lambda } \| \sigma^{(1)}(\lambda (u+u_b))(u+u_b)\|_0 \ .
\end{aligned}
\end{equation}
We notice that $\sigma^{(1)}$ is related to $G_k^{(2)}$ and it can be given in terms of the functions $A_n$ with $n=2$ defined in \eqref{def:Al}.
More explicitly, it takes the form  
\begin{equation}\label{def:sigma1}
\begin{gathered}
\sigma^{(1)}(\tilde{u})(v) = 
\frac{(-1)^{3}}{\lonenorm{f}} 2
\sum_{l=0}^2  
%\sum_{j\in P_l} \sum_{i\in P_{n-l}}
%\prod_{c=1}^l \prod_{\prod_{d=1}^{n-l}}
\int_{\mathcal{M}}\dd\mu_x \partial_kq_k(x)
\cdot \\
\left\{
(\Delta^U_R f)^l \otimes (\Delta^U_R f)^{2-l} \circ 
(1-\partial_\phi^2 \tilde{u} \Delta_R^U f) \otimes (1-\partial_\phi^2 \tilde{u} \Delta_R^U f)(w)(x,x)
\right\}
\partial_\phi^2 v\\
= A_2(\tilde{u})\partial_\phi^2 v.
\end{gathered}
\end{equation}
Thanks to the estimate given in Lemma \ref{lem:Al}, we have that 
\[
\begin{aligned}
\|\sigma^{(1)}(\lambda (u_b+u))(u_b+u)\|_0 
&\leq \|A_2\|_0 \|(u_b+u)\|_0 
\\
&\leq     
C'(1+\|u_b+u\|_2)   \|(u_b+u)\|_0 
\\
&\leq     
C'' \|u_b+u\|_2 
\end{aligned}
\]
%\[
%\|A_l(\tilde{u})\|_0\leq C(1+\|\tilde{u}\|_2) \ ,
%\]
for suitable constants $C'$ and $C''$ depending on $A$.
Using this estimate in \eqref{eq:stima-sigma}, 
and recalling the choices we made for $w$ in $\sigma(0)$, we obtain that for a suitable $c'$
\[
\sigma(u)  \geq  \sigma(0) - \sup_{\lambda } \| \sigma^{(1)}(\lambda (u+u_b))(u+u_b)\|_0 \geq  c' >0 \ ,
\]
thus concluding the proof.
\end{proof}

%
%\color{black}
%
%
%
%\begin{equation} \label{eq:RG for u}
%\partial_k u = G_k(\partial_\phi^2u)
%\end{equation}
%where the function $G_k$ is defined as 
%%%%\begin{equation}
%%%%G(k,\partial_\phi^2u) :=- \frac{1}{2}  \lim_{y\to x} \partial_k q_k(x) (1 -\Delta_R^u f \partial_\phi^2 u) w (1 - \Delta_A^u f \partial_\phi^2 u)(x,y)\ .
%%%%\end{equation}
%%\begin{equation}\label{eq:G}
%%G_k(\partial_\phi^2u) :=- \frac{1}{2 \lonenorm{f}}  \int_{\mathcal{M}} \dd\mu_x \partial_k q_k(x)  \lim_{y\to x} 
%%%\left\{ w(x,y)-  (\partial_\phi^2 u \Delta_R^u f \otimes 1)(w) (x,y) - (1\otimes 
%%%\partial_\phi^2 u \Delta_A^u f)(w)(x,y) + (\partial_\phi^2 u \Delta_R^u f \otimes \partial_\phi^2 u \Delta_A^u f)(w)(x,y)
%%%\right\}
%%\left\{(1-\partial_\phi^2 u \Delta_R^u f) \otimes (1-\partial_\phi^2 u \Delta_R^u f)(w)(x,y)
%%\right\}
%%\ .
%%\end{equation}
%\begin{equation}\label{eq:G}
%G_k(\partial_\phi^2u) :=- \frac{1}{2 \lonenorm{f}}  \int_{\mathcal{M}} \dd\mu_x \partial_k q_k(x)  
%%\lim_{y\to x} 
%%\left\{ w(x,y)-  (\partial_\phi^2 u \Delta_R^u f \otimes 1)(w) (x,y) - (1\otimes 
%%\partial_\phi^2 u \Delta_A^u f)(w)(x,y) + (\partial_\phi^2 u \Delta_R^u f \otimes \partial_\phi^2 u \Delta_A^u f)(w)(x,y)
%%\right\}
%\left\{(1-\partial_\phi^2 u \Delta_R^u f) \otimes (1-\partial_\phi^2 u \Delta_R^u f)(w)(x,x)
%\right\}
%\ .
%\end{equation}where $w \in C^{\infty}({\mathcal M}^2) $ is a given symmetric smooth function (the 
%

\color{blue}

\color{black}

%{\color{red} aggiustare 
%
%\begin{itemize}
%\item Say what is $\sigma$
%\item Chose $w$ in such a way that $\sigma(0)$ is positive.
%\item Restrict initial values and boundary condition in such a way that $\sigma(u_b)$ remains positive
%\item Chose $\mathcal{U}\subset F_0$ in such a way that $\sigma$ remain positive. 
%\item Say that $\sigma$ is smooth in $\partial_\phi^2 \tilde{u}$ and linear in $q_k$. 
%\item Say that $\sigma$ is tame.
%\end{itemize}
%
%}
%
%\begin{equation}
%L(u) v = \partial_k v - \frac{k^2}{2} \int_{x,y,z} \dot q_k(x) \Delta_{F,k,u}(x,y) \Delta_{F,k,u}(y,z)w(z,x) \partial^2_\phi v \ .
%\end{equation}

%\begin{equation}
%L(u) v = \partial_k v - \frac{i}{2} \int_\mathcal M \dd x \partial_k q_k(x) (P_{0,k} +  f k^2 \partial_\phi^2 u)^{-1}(x,z_1) \partial^2_\phi v (z_1,z_2) (P_{0,k} + f k^2 \partial_\phi^2 u )^{-1}(z_2,z_3) \lambda (z_3, x) \ .
%\end{equation}
%It can be computed using the definition of directional derivative on Fréchet spaces. 

%We can now define a distribution
%\[
%\psi_u(x) := k^2 \int_{y,z} \dot q_k(x) \Delta_{F,k,u}(x,y) \Delta_{F,k,u}(y,z) w(z,x)
%%
%%  \partial_k q_k(x) (P_{0,k} + f k^2 \partial_\phi^2 u )^{-1}(x,z_1) (P_{0,k} + f k^2 \partial_\phi^2 u )^{-1}(z_2, z_3) \lambda(z_3, x) \ .
%\]
%Since $q_k$ and $w$ are two compactly supported, smooth functions, and since the propagators $\Delta_{A,R}$ are never multiplied at a point of coincidence, $\sigma$ is a smooth function by theorems on distributions in Chapter 8 of H\"{o}rmander.
%\color{blue}
\begin{remark}
Notice that thanks to Proposition \ref{prop:select-sigma},  $\sigma$ can be chosen to be positive.
In applications to physics, when $w$ is obtained as the smooth part of the two-point function of a quantum state, it is not obvious that the choices necessary to have $\sigma$ positive can be made. In spite of this problem we observed in \cite{DDPR2022}  that this is the case in many physically sensible states, also thanks to the freedom in the split of the smooth part from the singular one present in any Hadamard two-point function. This freedom is related to the ordinary renormalization freedom when coinciding point limits are taken.
% 
%
% for physically sensible states is a positive function, and the convolutions with the classical M\o ller operators $(1+\Delta_{A,R} \U'')^{-1}$ preserve the sign of the function. The cut-off function $\dot q_k$ typically is proportionally to the mass scale $k$ times an IR cut-off function function $f$ which equals $1$ on the relevant portion of the spacetime, and so $\sigma$ is a positive function.
\end{remark}
%\color{black}
% We can say that $\psi_u$ is \textit{H\"{o}rmander integrable}.

Now we can prove the proposition
\begin{proposition} \label{prop:linearised RG operator tame smooth}
The linearisation of the RG operator
\[
L(u) v = \partial_k v -  \sigma \partial_\phi^2 v
\]
%\[
%L(u) v = \partial_k v - \frac{i}{2} \int_{\mathcal M \times \mathcal M} \psi_u(z_1, z_2) \partial_\phi^2 v(z_2,z_1)
%\]
is tame smooth.
\end{proposition}
\begin{proof}
Since $L$ acts as a second order linear differential operator, its $n-$th  order seminorm is controlled by the $n+2-$th order seminorm of $v$. 
%\color{blue}
Using the Lebiniz rule and an interpolating argument (see e.g. in Corollary 2.2.2 in \cite{Hamilton1982}) 
\[
\norm{L(u)v}_n \leq \norm{v}_{n+1} + C\left(  \|\sigma\|_0 \norm{v}_{n+2}+\|\sigma\|_{n+2} \norm{v}_{0}  \right)
%\leq C'(1+\|\sigma\|_{n+2})\norm{v}_{n+2} \ ,
\]
where $C$ is a constant. $\sigma$ is tame smooth and the composition of tame smooth maps is tame smooth, and thus $L:(\mathcal{U})\times F_0 \to F$ is tame smooth.
%\color{black}
%Since $\psi_u \partial_\phi^2 v $ is a H\"{o}rmander integrable distribution, we have
%\[
%\norm{L(u)v} \leq \norm{v} + \frac{1}{2} \int_{[a,b]} \dd k \int_{\mathcal M \times \mathcal M} \norm{ \psi_u(z_1, z_2)  \partial_\phi^2 v (z_2,z_1)} \leq \left ( 1 + \frac{\abs{b-a}}{2} C \norm{\psi_u} \right ) \norm{ \partial_\phi^2 v} \ ,
%\]
%where the constant $C$ depends on the volume of the intersection of the supports of $\psi_u$ and $ \partial_\phi^2 v$.
\end{proof}

\subsection{The linearisation of the RG operator is invertible, and the inverse is tame smooth}
%\color{blue} 
If $\sigma \geq c>0$ on $X\times [a,b]$,
%
%once suitable boundary conditions are fixed on $\partial X \times [a,b]$,
the linearised RG operator $L(u)$ on $X\times [a,b]$ has the form of a parabolic equation. 
The existence and uniqueness of an inverse which satisfies the chosen
boundary conditions
\[
E(g)(\phi,a) = 0,\qquad  E(g)|_{\partial X\times [a,b]}=0,  \qquad g\in C^\infty_0(X\times[a,b])
\]
is known \cite{Friedman}. Furthermore, by an application of the maximum principle, it is possible to prove that $E$ is continuous with respect to the uniform norm; see e.g. Section 3 of Chapter 2 in \cite{Friedman}. We collect these results in the following Proposition.

\begin{proposition} \label{prop: existence inverse E}
%$E(f)$ vanishes in the past and satisfies the chosen boundary condition.
Consider the linearised RG operator $L$. Assume that $\sigma(u)$ is positive for every $u\in \mathcal{U}\subset F$. 
Then, it exists an unique inverse $E:F \to F_0$  which is compatible with the initial and boundary conditions, thus satisfying 
\[
E(L(g))=L(E(g))=g, \qquad g\in F_0 \ .
\]
Moreover, the inverse is continuous with respect to the uniform norm. More precisely, it exists a positive constant $C>0$ such that
\[
\| E(g) \|_0 <  C   \| g\|_0.
\]
\end{proposition}
%
%\begin{proof}
%Since $\sigma$ is positive, $L$ is an ordinary parabolic equation. 
%The proof of the existence of the unique inverse can be found in \cite{Friedman}.
%The sought continuity can be proved as an application of the maximum principle.
%See e.g.    
%\end{proof}

We now pass to analyse the regularity of $E$, which is a necessary condition to apply the Nash-Moser Theorem. 

\begin{proposition} \label{prop: inverse is tame smooth}
%{\color{red} 
Consider the case where $\sigma\geq c > 0$, let $u\in \mathcal{U}\subset F_0 $ such that $\|u\|_4\leq A$, and assume that
$\sup_{i\in \{\phi,k\}} |D_i \log \sigma(u)|< \epsilon$ with a sufficiently small $\epsilon$.
%
%
% the initial conditions are such that $\|\psi\|_3+\|\beta\|_3 \leq \epsilon$ with a sufficiently small $\epsilon$ and that also $[a,b]$ is sufficiently small
%}
The inverse $E$ of the linearized RG operator $L$ is tame smooth.
\end{proposition}
\begin{proof}
We first observe that $L(u)$ depends on $u$ only through $\sigma$. Furthermore, $\sigma$ is a tame map of $u$. The composition of tame maps is tame, and so, to prove the statement, it suffices to study how $L$ depends on $\sigma$. To this end, with a little abuse of notation in this proof we shall denote $L(u)$ by $L(\sigma(u))$ and we estimate how $L$ depends on $\sigma$.
Consider 
$L(\sigma)(v)=g$. We look for an estimate which permits to control the higher derivative of $v$ with those of $g$. We start with two Lemmas.

\begin{lemma}\label{lemma:v1}
Under the hypothesis of Proposition \ref{prop: inverse is tame smooth},
% and if the interval $[a,b]$ is chosen to be sufficiently small, 
the following estimate holds.
% , $\sigma\geq c>0$ and $\|\sigma\|_2<A$, $\|\psi\|_1<\epsilon$ with a sufficietly small $\epsilon$.
\[
\|v\|_1 <  
C \left(
\|g\|_1
+
\| \sigma \|_1 \| g\|_0  
\right)
\]
\end{lemma}
\begin{proof}
The uniform continuity of $E$ stated in Proposition \ref{prop: existence inverse E} implies that if $L(\sigma)v=g$,
\[
\| v \|_0 < C \|g \|_0. 
\]  
We apply this continuity result to $D v$ where $D\in \{\partial_\phi, \partial_k\}$. 
We have
\[
\|D v\|_0 < C \|L(\sigma) D v\|_0.
\]
We observe that 
\begin{equation}\label{eq:LDV}
\begin{aligned}
L(\sigma)Dv &= D L(\sigma)v - D(\sigma) \partial_\phi^2 v \\
 &= D L(\sigma)v + \frac{D(\sigma)}{\sigma}\left( L(\sigma)v - \partial_k v  \right) 
\end{aligned}
\end{equation}
Hence, the uniform continuity of $E$ and the fact that $\sigma\geq c>0$ implies that
\begin{align*}
\|Dv\|_0  & < C \left(
\|D L(\sigma) v\|_0
+
\|D \log(\sigma) \|_0  (\| L(\sigma) v\|_0  + \|v\|_1 )
\right)
\\
& < C \left(
\|D g\|_0
+
\|D\log(\sigma) \|_0  (\| g\|_0  + \|v\|_1 )
\right)
\end{align*}
%hence, recalling that $\|\sigma\|_0> c$
%\begin{align*}
%\|Dv\|_0 &<C \left(
%\|D f\|_0
%+
%\| D \sigma\|_0  (\| f\|_0  + \|v\|_1 )
%\right)
%\end{align*}
considering all possible $D$, using the uniform continuity of $E$ and the fact that $1/\sigma > 1/c'$, we obtain
%if $C \| \sigma \|_1/c$ is smaller than 1 we have.
\begin{align*}
\|v\|_1 &\leq \left( \|v\|_0 + \sum_{D\in \{\partial_\phi,\partial_k\}} \|Dv\|_0\right)
<
C \left(
\|g\|_1
+
\| \sigma \|_1 \| g\|_0  
+ 
\sup_D\| D \log(\sigma) \|_0 
\|v\|_1 \right)
\end{align*}
hence
%considering all possible $D$ if $C \| \sigma \|_1/c$ is smaller than 1 we have.
\begin{align*}
(1-C \sup_D\| D \log(\sigma) \|_0)\|v\|_1 
<
C \left(
\|g\|_1
+
\| \sigma \|_1 \| g\|_0  
\right).
\end{align*}
Notice that by hypothesis, $ \sup_D\| D \log(\sigma) \|_0 \leq \epsilon$, 
hence, if $\epsilon $ is chosen sufficiently small
% D\log(\sigma)
%
%since $\|\sigma\|_2 < C\|u\|_4 \leq CA$,
%
% if both the time interval $(b-a)$ and $\|\psi\|_3+\|\beta\|_3<\epsilon$ are chosen to be sufficiently small  
\begin{equation}\label{eq:stima-brutta}
(1-C \sup_D\| D \log(\sigma) \|_0)\geq c'>0. 
\end{equation}
and 
%{\color{red} ATTENZIONE NON FUNZIONA QUANDO $u_b=-k\phi^2/2$} 
%
%To verify \eqref{eq:stima-brutta} we observe that 
%\[
% D \sigma(\phi,k) = D \sigma(\phi,a) +  \int_{a}^k \partial_kD\sigma(\phi,\kappa) \dd \kappa 
%\]
%hence, since both $\sigma$ and $D(\sigma)$ are smooth, we have that 
%\[
%\| D \sigma(\phi,k)\|_0 \leq  \| D \sigma(\phi,a)\|_0 +  (b-a)  \| \sigma\|_2 < C ( \epsilon + (b-a) (1+A)).
%\]
%where we used the fact that $\sigma$ is tame, and in particular $\|\sigma\|_{2} \leq C(1+\|u\|_4)\leq C(1+A)$. Furthermore, $\sigma(\phi,a)$ depends on $\phi$ and $a$ through $u_b+u$, hence in view of the continuity of $\sigma$,  
%\[
%| D \sigma(\phi,a)| \leq C| D\partial_{\phi}^2u_b(\phi,a) |\leq \|u_b\|_3\leq \epsilon
%\] 
%where we used the fact that $u=0$ at $k=a$ and the fact that we can chose $u_b$ in such a way that 
%$\|u_b\|_3\leq\|\psi\|_3+\|\beta\|_3 \leq \epsilon$.
%Hence if both $\epsilon$ and $b-a$ are sufficiently small $(1-C \sup_D\| D \sigma \|_0)\geq c'>0$.
%With this in mind we obtain
%
%
%{\color{red} CONTROLLARE CON LE CONDIZIONI AL CONTORNO }Hence
\begin{align*}
\|v\|_1 &< \frac{1}{(1-C \sup_D\| D \sigma \|_0)}
C \left(
\|g\|_1
+
\| \sigma \|_1 \| g\|_0  
\right)
\end{align*}
from which the thesis follows.
%\begin{align*}
%\|v\|_1 &<  
%C \left(
%\|f\|_1
%+
%\| \sigma \|_1 \| f\|_0  
%\right)
%\end{align*}
%
%\begin{align*}
%\| L(Dv) \|_0  &< \| D L(v)\|_0 + \| D(\log \sigma) \|_0 \left( \|  L(v)\|_0 + \|\partial_k v\|_0  \right) 
%\\
%  &< \| f\|_1 + \| D(\log \sigma) \|_0 \left( \|  f\|_0 + \| v\|_1  \right) 
%\end{align*}
%
%The case $n=0$ follows from 
\end{proof}

\begin{lemma}\label{lemma:vn}
Under the hypothesis of Proposition \ref{prop: inverse is tame smooth},
% and if the time interval is chosen to be sufficiently small, 
 it holds that for every $n$ 
\begin{equation}\label{eq:ineq-norm-L}
\|v\|_n < C\left(\|g\|_n +  \|g\|_0 \|\sigma\|_{n+1} \right)\ .
\end{equation}
\end{lemma}
\begin{proof}
We prove it by induction. 
The case $n=1$ follows from Lemma \ref{lemma:v1} and the standard property $\|\sigma\|_1 \leq \|\sigma\|_2$.
%
% \ref{prop: existence inverse E}, actually the uniform continuity of $E$ implies that
%\[
%\| v \|_0 < C \|f \|_0. 
%\]  
We assume now that inequality \eqref{eq:ineq-norm-L} holds up $n$. To prove that it holds also for the case $n+1$ we apply it to $D v$ where $D\in \{\partial_\phi, \partial_k\}$. We have
\begin{equation}\label{eq:inductive-step}
\|D v\|_n < C\left(\| L(\sigma)Dv \|_n +  \|L(\sigma)Dv\|_0 \|\sigma\|_{n+1} \right).
\end{equation}
Recalling \eqref{eq:LDV}, by Leibniz rule, the interpolating argument (Corollary 2.2.2 in \cite{Hamilton1982}) and the fact that $\sigma\geq c'>0$ we have 
\begin{align*}
\| L(\sigma)Dv \|_n  
%&<  \|g\|_{n+1} + \| \frac{D(\sigma)}{\sigma}  g\|_n + \| \frac{D(\sigma)}{\sigma} \partial_k v  \|_n
%\\
&<  \|g\|_{n+1} + \| {D(\log \sigma)}  g\|_n + \| D(\log \sigma) \partial_k v  \|_n
\\
&<  \|g\|_{n+1} +  C\left(\| \log{\sigma}\|_0  \|g\|_{n+1} + \| \log{\sigma}\|_{n+1}  \|g\|_0
+ \| \frac{D(\sigma)}{\sigma}\|_0\| \partial_k v  \|_n
\right.
\\
& \hspace{1cm}+
\left.\| \frac{D(\sigma)}{\sigma}\|_n \|\partial_k v  \|_0 \right)
\\
&<  C\left( \bigg (1+\| \log{\sigma}\|_0 \right) \|g\|_{n+1}  + \| \log{\sigma}\|_{n+1}  \|g\|_0
+ \| D(\log\sigma)\|_0 \| v  \|_{n+1} \\
&  \hspace{1cm}+
\| \frac{D(\sigma)}{\sigma}\|_n \|\partial_k v  \|_0 \bigg )
%\\
%&< \|f\|_{n+1} + \| \frac{D(\sigma)}{\sigma}\|_0  \|f\|_n + \| \frac{D(\sigma)}{\sigma}\|_n  \|f\|_0
%+ \| \frac{D(\sigma)}{\sigma}\|_0\| \partial_k v  \|_n
%+
%\| \frac{D(\sigma)}{\sigma}\|_n \|\partial_k v  \|_0 \ .
%\\
\end{align*}
From the last inequality, using the results of Lemma \ref{lemma:v1}, it thus follows that
\begin{equation}\label{eq:stepA}
\begin{aligned}
\| L(\sigma)Dv \|_n  
%&<  \left(1+\| \log{\sigma}\|_0 \right) \|g\|_{n+1}  + \| \log{\sigma}\|_{n+1}  \|g\|_0
%+ \| D(\log\sigma)\|_0 \| v  \|_{n+1}
%+
%\| \frac{D(\sigma)}{\sigma}\|_n \|v\|_1
%\\
&<  \left(1+\| \log{\sigma}\|_0 \right) \|g\|_{n+1}  + \| \log{\sigma}\|_{n+1}  \|g\|_0 \\
&+ \| D(\log\sigma)\|_0 \| v  \|_{n+1}
+
\| \frac{D(\sigma)}{\sigma}\|_n (\|g\|_1+ \|\sigma\|_1\|g\|_0)
%\\
%&< \|f\|_{n+1} + \| \frac{D(\sigma)}{\sigma}\|_0  \|f\|_n + \| \frac{D(\sigma)}{\sigma}\|_n  \|f\|_0
%+ \| \frac{D(\sigma)}{\sigma}\|_0\| \partial_k v  \|_n
%+
%\| \frac{D(\sigma)}{\sigma}\|_n \|\partial_k v  \|_0 \ .
%\\
\end{aligned}
\end{equation}
Furthermore from \eqref{eq:LDV} and Lemma \ref{lemma:v1}, we can prove that
\begin{equation}\label{eq:stepB}
\begin{aligned}
\|L(\sigma)Dv\|_0  &< 
\|D L(\sigma) v\|_0
+
\|D\sigma \|_0  (\| L(\sigma) v\|_0  + \|v\|_1 )
\\
& < 
\|D g\|_0
+
\|D\sigma \|_0  (\| g\|_0  + \|v\|_1 )
%\\
%& < 
%\|g\|_1
%+
%\|\sigma \|_1  (\| g\|_0  + \|v\|_1 )
\\
& < 
\|g\|_1
+
\|\sigma \|_1  ((1+\|\sigma\|_1)\|g\|_0  + \|g\|_1) \ .
\end{aligned}
\end{equation}
Hence, combining the two inequalities \eqref{eq:stepA} and \eqref{eq:stepB} in \eqref{eq:inductive-step}
\begin{gather*}
(1-C \| D(\log(\sigma))\|_0) \|v\|_{n+1} < C\left[
\left(1+\| \log{\sigma}\|_0 \right) \|g\|_{n+1}  + \| \log{\sigma}\|_{n+1}  \|g\|_0
+
\right.
\\
\left.
+
\| D(\log{\sigma})\|_n (\|g\|_1+ \|\sigma\|_1\|g\|_0)
+
\right.
\\
\left.
+\left(
\|g\|_1
+\|\sigma \|_1  (\| g\|_0  + \|\sigma\|_1\| g\|_0  + \|g\|_1)
\right)\|\sigma\|_{n+1}
\right]
\end{gather*}
Notice that, as stated in \eqref{eq:stima-brutta}, $(1-C \| D(\log(\sigma))\|_0)>0$, and so
\begin{gather*}
\|v\|_{n+1} < C\left[
\left(\left(1+\| \log{\sigma}\|_0 \right) \|g\|_{n+1}  + \| \log{\sigma}\|_{n+1}  \|g\|_0
+
\right.
\right.
\\
\left.
\left.
+
\| D(\log\sigma)\|_n (\|g\|_1+ \|\sigma\|_1\|g\|_0)
\right)
+
\right.
\\
\left.
+\left(
\|g\|_1
+
\|\sigma \|_1  ((1  + \|\sigma\|_1)\| g\|_0  + \|g\|_1 )
\right)\|\sigma\|_{ n+1}
\right] \ .
\end{gather*}
By the interpolating argument, it holds that $\|g\|_1\|h\|_{n} \leq C(\|g\|_0\|h\|_{n+1}+\|g\|_{n+1}\|h\|_{0})$. Moreover, we have
$\|\sigma\|_1 <A$ from which it follows that
$\|D\log\sigma\|_n \leq C \|\sigma\|_{n+1}$, and 
thus we obtain 
\begin{gather*}
\|v\|_{n+1} < C\left(\|g\|_{n+1}+ \|g\|_0 \|\sigma\|_{n+2}
\right)
\end{gather*}
and the thesis is proved.
%By continuity of $E$ we have
%\begin{align*}
%\|Dv\|_0 &< C \|L(\sigma) D v\|_0
%\\
%& < C \left(
%\|D L(\sigma) v\|_0
%+
%\|D(\log \sigma) \|_0  (\| L(\sigma) v\|_0  + \|v\|_1 )
%\right)
%& < C \left(
%\|D f\|_0
%+
%\|D(\log \sigma) \|_0  (\| f\|_0  + \|v\|_1 )
%\right)
%\end{align*}
%hence, recalling that $\|\sigma\|_0> c$
%\begin{align*}
%\|Dv\|_0 &<C \left(
%\|D f\|_0
%+
%\| D \sigma\|_0  (\| f\|_0  + \|v\|_1 )
%\right)
%\end{align*}
%considering all possible $D$ if $C \| \sigma \|_1/c$ is smaller than 1 we have.
%\begin{align*}
%\|v\|_1 &< C  \|v\|_0 + \sum_{D\in \{\partial_\phi,\partial_k\}} \|Dv\|
%<
%C \left(
%\|f\|_1
%+
%\| \sigma \|_1 \| f\|_0  
%+ 
%\sup_D\| D \sigma \|_0 
%\|v\|_1 \right)
%\end{align*}
%hence
%considering all possible $D$ if $C \| \sigma \|_1/c$ is smaller than 1 we have.
%\begin{align*}
%(1-C \sup_D\| D \sigma \|_0)\|v\|_1 &< C  \|v\|_0 + \sum_{D\in \{\partial_\phi,\partial_k\}} \|Dv\|
%<
%C \left(
%\|f\|_1
%+
%\| \sigma \|_1 \| f\|_0  
%\right)
%\end{align*}
%Notice that since $\|\sigma\|_2 < A$, if the time interval is chosen to be sufficiently short 
%$(1-C \sup_D\| D \sigma \|_0)\geq c'>0$. Hence
%\begin{align*}
%\|v\|_1 &< \frac{1}{(1-C \sup_D\| D \sigma \|_0)}
%C \left(
%\|f\|_1
%+
%\| \sigma \|_1 \| f\|_0  
%\right)
%\end{align*}
%from which
%\begin{align*}
%\|v\|_1 &<  
%C \left(
%\|f\|_1
%+
%\| \sigma \|_1 \| f\|_0  
%\right)
%\end{align*}
%
%
%
%
%%
%%\begin{align*}
%%\| L(Dv) \|_0  &< \| D L(v)\|_0 + \| D(\log \sigma) \|_0 \left( \|  L(v)\|_0 + \|\partial_k v\|_0  \right) 
%%\\
%%  &< \| f\|_1 + \| D(\log \sigma) \|_0 \left( \|  f\|_0 + \| v\|_1  \right) 
%%\end{align*}
%%
%
%The case $n=0$ follows from 
\end{proof}
Estimates of Lemma \ref{lemma:vn} implies that $E f$ is a tame map of $\sigma$ and $f$.
The map $\sigma(u)$ is a smooth tame function of $u$. The composition of tame maps is tame, and so we have the result. 
\end{proof}

To prove that $E$ is tame smooth we made two assumptions for $\sigma$: first, that $\sigma>c$, and second, that 
$\partial_{i} \log(\sigma) <\epsilon'$ for small $\epsilon'$.
We have already seen in Proposition \ref{prop:select-sigma} that $u_b$ and $\mathcal{U}$ can be chosen in such a way that, for every $u\in \mathcal{U}$,
$\sigma(u_b+u)>c$.
We now want to prove that the second requirement also holds.

\begin{proposition}
Let $\epsilon'>0$ and consider the initial conditions given in \eqref{eq:initial-boundary-condition}, and let
$u_b$ in \eqref{eq:ub} be such that  $\|u_b\|_3\leq A$.
If $[a,b]$ is such that $b-a$ together with $A$ are sufficiently small,
%and $b$ is sufficiently large, 
%or if $m$ in $P_0$ is sufficiently large,  
it holds that 
\[
|\partial_{i} \log(\sigma)| < \epsilon' 
\]
for every $u\in \mathcal{U}$, recalling that $\|u\|_4<A$.
\end{proposition}
\begin{proof}
%We study
%\[
%\partial_{i} \log(\sigma) <\epsilon
%\]
Let $D$ be either $\partial_\phi$ or $\partial_k$, and
notice that $D\log \sigma = D\sigma /\sigma$. In proposition \ref{prop:select-sigma} we have shown that there are choices of $w$ for which $1/\sigma< 1/c$.
%To verify \eqref{eq:stima-brutta} 
We now observe that 
\[
 D \sigma(\phi,k) = D \sigma(\phi,a) +  \int_{a}^k \partial_{\kappa}D\sigma(\phi,\kappa) \dd \kappa \ .
\]
Therefore, since both $\sigma$ and $D(\sigma)$ are smooth, we have that 
\[
\| D \sigma(\phi,k)\|_0 \leq  \| D \sigma(\phi,a)\|_0 +  (b-a)  \| \sigma\|_2 < C ( A + (b-a) (1+A)) \ ,
\]
where we used the fact that $\sigma$ is tame, and in particular $\|\sigma\|_{2} \leq C(1+\|u\|_4)\leq C(1+A)$. Furthermore, $\sigma(\phi,a)$ depends on $\phi$ and $a$ through $u_b+u$; hence, in view of the continuity of $\sigma$,  
\[
| D \sigma(\phi,a)| \leq C| D\partial_{\phi}^2u_b(\phi,a) |\leq \|u_b\|_3\leq A
\] 
where we used the fact that $u=0$ at $k=a$ and the fact that we can chose $u_b$ in such a way that 
$\|u_b\|_3<A$.
%$\|u_b\|_3\leq\|\psi\|_3+\|\beta\|_3 \leq \epsilon$.
Therefore,  
\[
|D \log\sigma | =  \frac{|D \sigma(\phi,k)|}{\sigma}\leq \frac{C}{c} ( A + (b-a) (1+A)) \leq \epsilon' \ ,
\]
where we have chosen both $A$ and $b-a$ sufficiently small to make the last inequality valid. 
\end{proof}

\begin{remark}
We recall that in $u_b$ there is a contribution which is $-q_k \phi^2$. For  $q_k$ given in \eqref{eq:qk} it is in general not possible to make the choice $\|q_k \phi^2\|_3\leq A$ for small $A$ because of the constant contribution $k_0 f$  in \eqref{eq:qk}, while the other corrections can be made small with judicious  choices of the chosen parameters. However, as observed above, such a contribution can always be reabsorbed in the mass of the free theory present in $P_0$.
%
%moved to the mass of the  
%
%$k_0 f$ can always be 
%
%We however obvious that 
%
%have seen in
%
% \def{q_k}
%vv
\end{remark}

\begin{theorem}\label{theorem:existence of solutions}
Under the hypothesis of Proposition \ref{prop: inverse is tame smooth},
the RG operator admits a unique family of tame smooth local inverses, and unique local solutions of the RG flow equations exist.
\end{theorem}
\begin{proof}
The proof is a direct application of the Nash-Moser theorem \cite{Hamilton1982}, which can be applied thanks to the results of Propositions  \ref{prop: RG tame smooth}, \ref{prop:linearised RG operator tame smooth}, \ref{prop: existence inverse E}, and \ref{prop: inverse is tame smooth}.
Actually, it follows from the Nash-Moser theorem that the RG operator admits a unique family of tame smooth local inverses. This guarantees the existence of local solutions of the RG flow equations.
\end{proof}

%\color{black}

%\section{Conclusions}

\section{Acknowledgments}
E.D. is supported by a PhD scholarship of the University of Genoa and by the project GNFM-INdAM Progetto Giovani \textit{Non-linear sigma models and the Lorentzian Wetterich equation}, CUP\textunderscore E53C22001930001.
E.D. and N.P. are grateful for the support of the National Group of Mathematical Physics (GNFM-INdAM).
The research performed by N.P. and E.D. was supported in part by the MIUR Excellence Department Project 2023-2027 awarded to the Dipartimento di Matematica of the University of Genova, CUP\textunderscore D33C23001110001.

We are grateful to Nicolò Drago and Kasia Rejzner for many useful discussions and a thorough reading of the manuscript.

%\begin{proof}
%From the explicit expression for the inverse $F$, we can evaluate the inverse $E$ from its perturbation series. First, notice that
%\[
%E = \sum_{n= 0}^\infty (-FV)^n F = F - F VF + (-1)^2 F (V F V F ) + ... = \sum_{n= 0}^\infty F (- VF )^n \ .
%\]
%On the other hand,
%\[
%\partial_\phi F = - \frac{ \sigma}{4 \sqrt{\bar \sigma}} \int \frac{1}{\sqrt{\bar \sigma} } \dd \phi F = M F \ ,
%\]
%where $M$ acts as multiplication. By the Gr\"onwall lemma in differential form, $F$  satisfies
%\[
%F(\phi) \leq F(\phi_0) \exp{\int_{\phi_0}^{\phi} M \dd \phi' } \ ,
%\]
%while
%\[
%VF (\phi)\leq \partial_\phi \sigma F(\phi_0) \exp{\int_{\phi_0}^{\phi} M \dd \phi' } \ .
%\]
%and so
%\[
%E \leq \sum_{n= 0}^\infty F(\phi_0)^{n+1} \exp{(n+1)\int_{\phi_0}^{\phi_1} M \dd \phi} (-M)^n \ .
%\]
%The r.h.s. of the above inequality acts simply as multiplication, so the inverse $E$ satisfies in norm
%\[
%\norm{E v}_n \leq C( 1 + \norm{v}_n) \ .
%\]
%\end{proof}

\printbibliography
\end{document}